\documentclass[sigconf]{acmart}
\usepackage{subfig}
\usepackage{booktabs}
\usepackage{multirow}
\usepackage{amsmath}
\usepackage{amsfonts}
\usepackage{algorithm}
\usepackage{algorithmic}
\usepackage{lipsum}  
\usepackage{bm}
\usepackage{soul}
\usepackage{xcolor}
\usepackage[normalem]{ulem}
\usepackage{xifthen}
\usepackage[switch]{lineno}
\usepackage{float}
\useunder{\uline}{\ul}{}

\DeclareMathOperator*{\argmin}{arg\,min}
\newcommand{\ith}{\ensuremath{i^{\text{th}}}}
\newcommand{\jth}{\ensuremath{j^{\text{th}}}}
\newcommand{\eigv}[1][]{\ensuremath{\ifthenelse{\isempty{#1}}{\bm{q}} {\bm{q}^{(#1)}}}}
\newcommand{\param}[1][]{\ensuremath{\ifthenelse{\isempty{#1}} {\bm{w}} {\bm{w}_{#1}}}}
\newcommand{\sysname}{\textsc{Leporid}}
\newcommand{\modname}{\textsc{DLR$^2$}}

\AtBeginDocument{%
  \providecommand\BibTeX{{%
    \normalfont B\kern-0.5em{\scshape i\kern-0.25em b}\kern-0.8em\TeX}}}
    
\usepackage{lipsum}

\setcopyright{none}

\begin{document}
\fancyhead{}
\widowpenalty=0
\clubpenalty=0

\title{Initialization Matters: Regularizing Manifold-informed Initialization for Neural Recommendation Systems}
\author{Yinan Zhang$^{1,3}$, Boyang Li$^{1,*}$, Yong Liu$^{2,3}$, Hao Wang$^4$, Chunyan Miao$^{1,2,3,*}$}\thanks{$^*$ Corresponding Author}
\affiliation{
$^1${School of Computer Science and Engineering, Nanyang Technological University}\\
$^2${Joint NTU-UBC Research Centre of Excellence in Active Living for the Elderly (LILY)}\\
$^3${Alibaba-NTU Singapore Joint Research Institute}
$^4${Alibaba Group}
}
\email{
{yinan002, boyang.li, stephenliu, ascymiao}@ntu.edu.sg, 
cashenry@126.com
}
% \fi
% \author{
% Yinan Zhang$^{1,3}$\and
% Boyang Li$^{3}$\footnote{Contact Author}\and
% Yong Liu$^{1,2}$\and
% Hao Wang$^4$\And
% Chunyan Miao$^{1,2,3}$\footnote{Contact Author}\\
% \affiliations
% $^1$Alibaba-NTU Singapore Joint Research Institute\\
% $^2$Joint NTU-UBC Research Centre of Excellence in Active Living for the Elderly (LILY)\\
% $^3$School of Computer Science and Engineering, Nanyang Technological University\\
% $^4$Alibaba Group\\
% \emails
% yinan002@e.ntu.edu.sg,
% boyang.li@ntu.edu.sg,
% stephenliu@ntu.edu.sg,
% ascymiao@ntu.edu.sg
% % \{first, second\}@example.com,
% % third@other.example.com,
% % fourth@example.com
% }



\begin{abstract}
  Proper initialization is crucial to the optimization and the generalization of neural networks. However, most existing neural recommendation systems initialize the user and item embeddings randomly. 
  In this work, we propose a new initialization scheme for user and item embeddings called Laplacian Eigenmaps with Popularity-based Regularization for Isolated Data (\sysname). \sysname{} endows the embeddings with information regarding multi-scale neighborhood structures on the data manifold and performs adaptive regularization to compensate for high embedding variance on the tail of the data distribution. Exploiting matrix sparsity, \sysname{} embeddings can be computed efficiently. 
  We evaluate \sysname{} in a wide range of neural recommendation models. In contrast to the recent surprising finding that the simple $K$-nearest-neighbor (KNN) method often outperforms neural recommendation systems \cite{recsys19}, we show that existing neural systems initialized with \sysname{} often perform on par or better than KNN. To maximize the effects of the initialization, we propose the Dual-Loss Residual Recommendation (\modname) network, which, when initialized with \sysname{}, substantially outperforms both traditional and state-of-the-art neural recommender systems. 
%   \blfootnote{$^\ast$ Corresponding Author}
\end{abstract}





\maketitle

\section{Introduction}
%Deep neural networks have achieved sweeping success in numerous areas of artificial intelligence and have been also applied to recommendation systems.

% initial introduction
Deep neural networks have been widely applied in recommendation systems. However, the recent finding of \cite{recsys19} indicates that the simplistic $K$-nearest-neighbor (KNN) method can outperform many sophisticated neural networks, casting doubt on the effectiveness of deep learning for recommendation. 
The power of KNN lies in its use of the neighborhood structures on the data manifold. As neural networks tend to converge close to their initialization (Theorems 1 \& 2 of \cite{allenzhu2018convergence}; Collorary 4.1 of \cite{du2019gradient}), a natural thought is to initialize neural recommendation networks with neighborhood information.
% initial introduction

% % yinan modified
% Deep neural networks have been widely applied in recommendation systems. However, the recent finding of \cite{recsys19} indicates that the simplistic $K$-nearest-neighbor (KNN) method can outperform many matrix factorization methods (e.g., Bayesian Personalized Ranking (BPR) \cite{bpr}) and sophisticated neural networks, casting doubt on the effectiveness of deep learning for recommendation. 
% The power of KNN lies in its use of the neighborhood structures on the data manifold. As neural networks tend to converge close to their initialization (Theorems 1 \& 2 of \cite{allenzhu2018convergence}; Collorary 4.1 of \cite{du2019gradient}), a natural thought is to initialize neural recommendation networks with neighborhood information, where most existing methods only utilize random or BPR initializations.
% % yinan modified

Laplacian Eigenmaps (LE) \cite{initEigenmaps} is a well-known technique for creating low-dimensional embeddings that capture multi-scale neighborhood structures on graphs. In Figure \ref{fig:multi-scale-neighborhood}, we visualize how the first four dimensions of LE embeddings divide the graph into clusters of increasingly finer granularity. For example, the positive and negative values in the second dimension partition the graph into two (blue vs. red), whereas the third dimension creates three partitions. This suggests that LE embeddings can serve as an effective initialization scheme for neural recommendation systems. Compared to attribute-based node feature pretraining for graphic neural networks \cite{hu2020strategies,gpt_gnn2020,zhang2020graphbert}, the LE initialization captures the intrinsic geometry of the data manifold, does not rely on external features such as user profiles or product specifications, and is fast to compute.

% an example of LE embeddings for graph nodes. 
% All the values of the first eigenvector are positive, while the second eigenvector shows two clusters and so on. We can conclude that the first few dimensions of the embeddings indicate large neighborhoods whereas later dimensions indicate increasingly smaller neighborhood structures. 

\begin{figure*}[t]
\centering
\includegraphics[width=\textwidth]{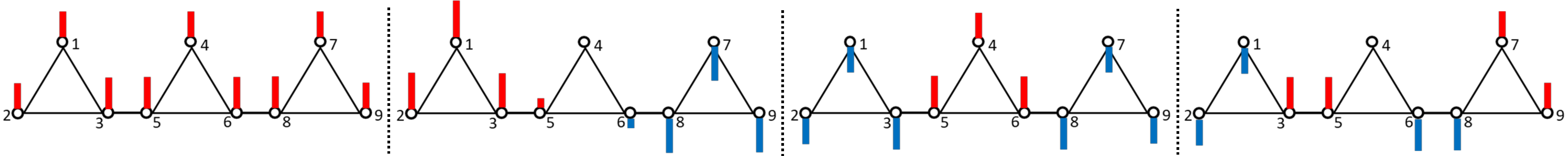}
\caption{Node embeddings generated by Laplacian Eigenmaps (LE). We show the first four embedding dimensions sequentially (left to right) as bars overlaid on every node. Positive/negative values are denoted by red/blue, and the bar lengths indicate the absolute values. We observe greater fluctuation between neighboring nodes in later dimensions.}
\label{fig:multi-scale-neighborhood}
\end{figure*}

Nevertheless, direct application of LE ignores an important characteristic of the KNN graph, that the graph edges are estimated from different amounts of data and hence have unequal variance. If we only observe short interaction history for some users and items, it is difficult to obtain good estimates of their positions (i.e. neighbors) on the data manifold and of their embeddings. As a result, direct application of LE may lead to poor performance. 

% initial 
% An important characteristic of the KNN graph is that the edges are estimated with different amount of data and hence have unequal variance. When items and users have short interaction history, our empirical estimates of their neighboring relations are estimated with little data and have high variance. As a result, their LE embeddings also have high variance. When new data come in, the embeddings may change drastically. Figure \ref{fig:embedding-change} illustrates this phenomenon using simulation on random graphs. Therefore, naive application of LE may lead to poor performance for users and items on the tail of the data distribution. 

To address the problem, we propose Laplacian Eigenmaps with Popularity-based Regularization for Isolated Data (\sysname). \sysname{} regularizes the embeddings based on the observed frequencies of users and items, serving as variance reduction for estimated embeddings. \sysname{} embeddings can be efficiently solved by using sparse matrix eigendecomposition with known convergence guarantees. Experimental results show that \sysname{} outperforms existing initialization schemes on a wide range of neural recommender systems, including collaborative filtering, graph neural networks, and sequential recommendation networks. Although randomly initialized neural recommendation methods compare unfavorably with KNN techniques, with \sysname{} initializations, they often perform on par or better than KNN methods. 

To fully realize the strength of \sysname{}, we propose Dual-Loss Residual Recommendation (\modname{}), a simple residual network with two complementary losses. On three real-world datasets, \modname{} with \sysname{} initialization outperforms all baselines by large margins. An ablation study reveals that the use of residual connections contributes significantly to the performance of \modname{}. Under certain conditions, the output of residual networks contains all information in the input \cite{behrmann2019invertible}. Thus, the proposed architecture may utilize manifold information throughout the network, improving performance. 

\sloppypar{With this paper, we make the following contributions:
\begin{itemize}
    \item Consistent with \cite{recsys19}, we find that randomly initialized neural recommenders perform worse than KNN methods. However, with \sysname{} initialization, neural recommenders outperform well-tuned KNN methods, underscoring the importance of initialization for neural recommender systems. 
    \item To capture neighborhood structures on the data manifold, we propose to initialize user and item representations with Laplacian Eigenmaps (LE) embeddings. To the best of our knowledge, this is the first work applying LE in initializing neural recommendation systems.  
    \item Further, we identify the issue of high variance in LE embeddings for KNN graphs and propose Popularity-based Regularization for Isolated Data (\sysname) to reduce the variance for users and items on the tail of the data distribution. The resulted \sysname{} embeddings outperform LE embeddings on all neural recommendation systems that we tested.
    \item We propose a new network architecture, \modname{}, which employs residual connections and two complimentary losses to effectively utilize  the superior initialization. On three real-world datasets, \modname{} initialized by \sysname{} surpasses all baselines including traditional and proper initialized deep neural networks by large margins. 
\end{itemize}}

\section{Related Work}
% In this section, we first discuss the importance of neural network initialization, followed by the existing recommendation methods. 
% Finally, we review sequential recommendation models that capture users' long-term and short-term interests.
\subsection{Initializing Neural Networks}
Proper initialization is crucial for the optimization and the generalization of deep neural networks. Poor initialization was a major reason that early neural networks did not realize the full potential of first-order optimization \cite{Sutskever2013:important_init}. 
Xavier \cite{glorot2010:init} and Kaiming initializations \cite{He_2015_ICCV} are schemes commonly used in convolutional networks. 
Further highlighting the importance of initialization, recent theoretical developments (Theorems 1 \& 2 of \cite{allenzhu2018convergence}; Collorary 4.1 of \cite{du2019gradient}) show the point of convergence is close to the initialization. Large pretrained models such as BERT \cite{devlin2018bert} can also be seen as supplying good initialization for transfer learning. 

While most neural recommendation systems make use of random initialization, some recent works  \cite{ncf,bpr1,bpr2} initialize embeddings from BPR \cite{bpr}. \cite{Zhu2020} embeds user and item representations using node2vec \cite{node2vec} on the heterogeneous graph containing both explicit feedback (e.g., ratings) and side information (e.g., user profiles). \cite{gpt_gnn2020,zhang2020graphbert} explore pretraining for graph networks, but they critically rely on side information like visual features of movies.

In comparison, LE and \sysname{} initialization schemes represent the intrinsic geometry of the data manifold and do not rely on any side information. To the best of our knowledge, this is the first work employing intrinsic geometry based initialization for neural recommendation systems. 
Experiments show that the proposed initialization schemes outperform a wide range of existing ones such as pretrained embeddings from BPR \cite{bpr} and Graph-Bert \cite{zhang2020graphbert} without side information. In addition, LE and \sysname{} initializations can be efficiently obtained with established convergence guarantees (see Section \ref{sec:time-complexity}), whereas large pretrained neural networks like \cite{zhang2020graphbert} require significantly more computation and careful tuning of hyperparameters for convergence.

\subsection{Recommendation Systems}
From the aspect of research task, sequential recommendations \cite{seq-wu18,seq-xu19,caser,seq-wang18,seq-yuan19,seq-wu17,seq-bal16,Steffen19} are recently achieving more research attentions, as the interactions between users and items are essentially sequentially dependent \cite{seq-wang19}. Sequential recommendations explicitly model  users' sequential behaviors.

% Existing recommendation methods are mainly based on collaborative filtering \cite{bpr,itemknn,userknn} and deep learning \cite{ncf,bpr1}. Sequential recommendation systems \cite{caser,seq-yuan19,seq-xu19,seq-bal16,seq-bal17,seq-wang19} are recently achieving more research attention, which explicitly model the sequential nature of the user-item interactions.

% Most traditional recommendation methods are based on collaborative filtering (CF) \cite{bpr,itemknn,userknn}. Recent years also have witnessed the rapid development of deep learning models \cite{ncf,bpr1}. We can divide recommendation methods into sequential methods and unsequential ones. Sequential recommendation methods \cite{caser,seq-yuan19,seq-xu19,seq-bal16,seq-bal17} explicitly model the sequential nature of the user-item interactions while the latter ones \cite{ncf,NGCF19} assume all interaction data are of the same importance.
% Since our initialization \sysname{} utilize the KNN graph information, we will briefly review existing graph-based recommendation methods. We will then introduce sequential recommendation methods since \modname{} captures users' sequential behavior.

From the aspect of recommendation techniques, existing methods are mainly based on collaborative filtering \cite{bpr,itemknn,userknn} and deep learning \cite{ncf,bpr1,jin2020efficient}. 
Graph Neural Networks (GNN) have recently gained popularity in recommendation systems \cite{graph-rex18,graph-qu20,graph-chang20}. 
% Traditional graph-based methods \cite{graph-konstas09,graph-lee11} perform Personalized PageRank \cite{personalized-page-rank:2003}, and recommend graph nodes based on their probabilities in the stationary distribution. 
GNN-based recommendation methods \cite{wu2020graph} typically build the graph with users and items as nodes and interactions as edges, and apply aggregation methods on the graph. Most works \cite{graph-chen20,graph-wang20,graph-qin20,graph-song19} make use of spatial graph convolution, which focuses on local message passing. A few works have investigated the use of spectral convolution on the graph \cite{spectral_Zheng_2018,spectral-Liu2019}.

% Traditional graph-based techniques \cite{graph-konstas09,graph-lee11} perform Personalized PageRank \cite{personalized-page-rank:2003}, and recommend graph nodes based on their probabilities in the stationary distribution. In recent years, Graph Neural Networks (GNN) have gained popularity in recommender systems \cite{graph-rex18,graph-qu20,graph-chang20}, where the graph contains users and items as nodes and their interaction as edges. Most GNN-based approaches \cite{graph-song19,graph-chen20,graph-wang20,graph-qin20} make use of spatial graph convolution, which focuses on local message passing. A few works have investigated the use of spectral convolution on the graph \cite{spectral_Zheng_2018,spectral-Liu2019}. 

The initialization schemes we propose are orthogonal to the above works on network architecture. Empirical results demonstrate that the LE and \sysname{} initialization schemes boost the accuracy of a diverse set of recommendation systems, including GNNs that already employ graph information.

\section{Proposed Initialization Schemes}

% \iffalse
% % \hl{Needs a table of notations}
\begin{table}[t]
    \renewcommand{\arraystretch}{1.2}
\centering
\caption{Frequently used notations}
\begin{tabular}{lp{18.5em}}
\toprule
$\mathcal{D}$ & Dimensionality of user and item embeddings \\
$\mathcal{N}$ & Number of nodes in the graph\\
% $M$ & Total number of users \\
% $N$ & Total number of items \\
$G$ & Graph used in Laplacian eigenmaps\\
$W$ & Weighted adjacency matrix \\
$W_{i,j}$ & Weight of the edge between nodes $i$ and $j$ \\
$d_i$ & Degree of node $i$\\
$d_{\text{max}}$ & Maximum degree of the graph\\
$D$ & Diagonal degree matrix \\
$L$ & Unnormalized graph Laplacian \\
$L_{\text{sym}}$ & Normalized symmetric graph Laplacian \\
$L_{\text{reg}}$ & Unnormalized graph Laplacian with popularity-based regularization \\
$L_{\text{regsym}}$ & Normalized symmetric graph Laplacian with popularity-based regularization \\
$\lambda_i$ & $\ith$ eigenvalue of the graph Laplacian\\
$\bm{q}^{(i)}$ & $\ith$ eigenvector of the graph Laplacian\\
% $\mathcal{I}_{i}$ & Set of items that user $i$ interacted with\\
\bottomrule
\end{tabular}
\end{table}
% \fi

\subsection{Laplacian Eigenmaps}
Given an undirected graph $G=\langle V, E \rangle$ with non-negative edge weights, Laplacian Eigenmaps creates $\mathcal{D}$-dimensional embeddings for the $\mathcal{N}$ graph nodes. We use $W \in \mathbb{R}^{\mathcal{N} \times \mathcal{N}}$ to denote the adjacent matrix, whose entry $W_{i,j}$ is the weight of the edge between nodes $i$ and $j$. $W_{i,j} \ge 0$ and $W_{i,j} = W_{j,i}$. If nodes $i$ and $j$ are not connected, $W_{i,j} = W_{j,i} = 0$. The degree of node $i$, $d_i = \sum_j W_{i,j}$, is the sum of the weights of edges adjacent to $i$. 
The degree matrix $D \in \mathbb{R}^{\mathcal{N} \times \mathcal{N}}$ has $d_i$ on its diagonal, $D_{i,i}=d_i$, and zero everywhere else. The unnormalized graph Laplacian is defined as
\begin{equation}
    L = D-W.
\end{equation}
$L$ is positive semi-definite (PSD) since for all $\bm{q}\in \mathbb{R}^{\mathcal{N}}$,
\begin{equation}
    \bm{q}^\top L \bm{q} = \frac{1}{2}  \sum_{i}\sum_{j} W_{i,j} (q_i - q_j)^2 \ge 0,
\end{equation}
where $q_i$ refers to the $i^{th}$ component of $\bm{q}$.
Alternatively, we may use the normalized symmetric graph Laplacian $L_{\text{sym}}$ \cite{chung97, andrew02},
\begin{equation}
\label{eq:l_sym}
    L_{\text{sym}} = D^{-\frac{1}{2}}(D-W)D^{-\frac{1}{2}}.
\end{equation}
To find node embeddings, we compute the eigendecomposition $L = Q\Lambda Q^\top$,
% \hl{
where the column of matrix $Q$ are the eigenvectors $\eigv[1], \ldots, \eigv[\mathcal{N}]$ and matrix $\Lambda$ has the eigenvalues $0 = \lambda_1 \le \lambda_2 \le \ldots \le \lambda_{\mathcal{N}}$ on the diagonal. We select the eigenvectors corresponding to the smallest $\mathcal{D}$ eigenvalues as matrix $\Tilde{Q} \in \mathbb{R}^{\mathcal{N}\times \mathcal{D}}$. The rows of $\Tilde{Q}$ become the embeddings of the graph nodes. Typically we set $\mathcal{D} \ll \mathcal{N}$. 
% }

We can think of $q^{(k)}_{j}$, the $\jth$ component of $\eigv[k]$, as assigning a value to the graph node $j$. When $\lambda_k$ is small, for adjacent nodes $i$ and $j$, the difference between $q^{(k)}_{i}$ and $q^{(k)}_{j}$ is small. As $\lambda_k$ increases, 
the $k^{\text{th}}$ embedding dimension vary more quickly across the graph (see Figure \ref{fig:multi-scale-neighborhood}). The resulted embeddings represent multi-scale neighborhood information. 

With $L$ being PSD, LE can be understood as solving a sequence of minimization problems. 
The eigenvectors $\eigv[1], \ldots, \eigv[\mathcal{D}]$ are solutions of the constrained minimization 
\begin{equation}
    \eigv[k] = \argmin_{\eigv} \frac{1}{2} \sum_{i}\sum_{j} W_{i,j} (q_{i} - q_{j})^2 \\
\label{eq:le-optimization}
\end{equation}
\begin{equation}
\text{s.t. \, } \eigv^\top \eigv = 1 \text{\, and \,} \bm{q}^\top \eigv[l] = 0, \forall l < k.
\label{eq:le-optimization-constraints}
\end{equation}
We emphasize that, different from the optimization of non-convex neural networks, the LE optimization can be solved efficiently using eigendecomposition with known convergence guarantees (see Section \ref{sec:time-complexity}).

\subsection{Popularity-based Regularization}

We generate user embeddings from the KNN graph where nodes represent users and two users are connected if one is among the K nearest neighbors of the other. Item embeddings are generated similarly. An important observation is that we do not always have sufficient data to estimate the graph edges well. For nodes on the tail of the data distributions (users or items in nearly cold-start situations), it is highly likely that their edge sets are incomplete. 

We empirically characterize the phenomenon by generating random Barabasi-Albert graphs \cite{BA-graph2002} and adding one extra edge to nodes with different degrees. We then compute the change in embedding norms caused by the addition, which is averaged over 50 random graphs. For every node on a graph, we randomly pick one edge from all possible new adjacent edges and add it to the graph. This procedure is repeated 30 times for every node. Figure \ref{fig:embedding-change} plots the change in norm versus the degree of the node before the addition. Graph nodes with fewer adjacent edges experience more drastic changes in their LE embeddings. This indicates that the LE embeddings for relatively isolated nodes are unreliable because a single missing edge can cause large changes.

\begin{figure}[t]
\centering
\includegraphics[width=0.7\columnwidth]{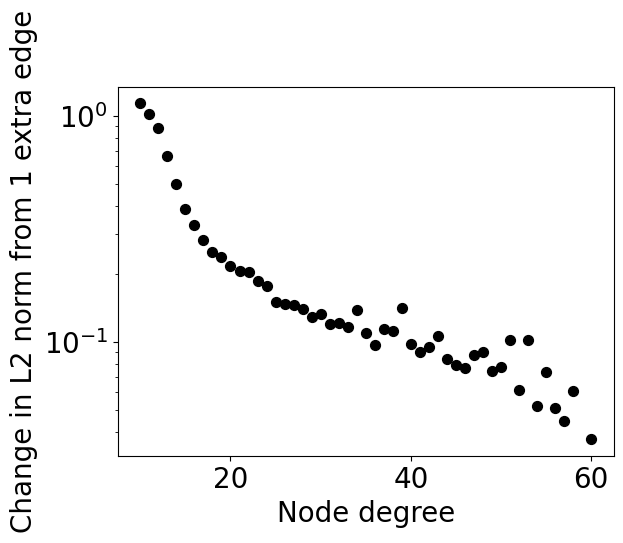}
\caption{Graph nodes with fewer adjacent edges experience more drastic changes in their LE embeddings when one adjacent edge is added to the node. The y-axis shows the change in the $\ell_2$ norm of the embedding of the node. The x-axis shows the node degree before the addition. %Results are averaged over 50 Barabási–Albert random graphs \cite{BA-graph2002} and 30 independent additions on each node.
}
\label{fig:embedding-change}
\end{figure}

As a remedy, we propose \sysname{}, which selectively regularizes the embeddings of graph node $i$ when $d_i$ is small. 
Keeping the constraints in Eq. \ref{eq:le-optimization-constraints} unchanged, we replace Eq. \ref{eq:le-optimization} with
\begin{equation}
%\begin{split}
    \eigv[k] = \argmin_{\eigv} \frac{1}{2} \sum_{i} \sum_{j} W_{i,j} (q_{i} - q_{j})^2 + \alpha \sum_{i} (d_{\text{max}} - d_i) q_i^2 ,
%\end{split}
\label{eq:lepor-optimization}
\end{equation}
where $d_{\text{max}} = \max_i d_i$ is the maximum degree of all nodes and $\alpha$ is a regularization coefficient. The $(d_{\text{max}} - d_i) q_i^2$ term applies stronger regularization to $q_i$ as $d_i$ decreases. In matrix form,
\begin{equation}
%\begin{split}
    \eigv[k] = \argmin_{\eigv} \eigv^\top L_{\text{reg}} \eigv 
    = \argmin_{\eigv} \eigv^\top ((1-\alpha)D-W+\alpha d_{\text{max}} I) \eigv .
%\end{split}
\label{eq:lepor-optimization-matrix}
\end{equation}
As $d_{\text{max}} \ge d_i$, $L_{\text{reg}}$ remains PSD, so we can solve the sequential optimization using eigendecomposition as before. For the normalized Laplacian, with $D_{\text{reg}} = (1-\alpha)D+\alpha d_{\text{max}} I$, we perform eigendecomposition on  $L_{\text{regsym}}$,
\begin{equation}
\label{eq:l_regsym}
L_{\text{regsym}} = D_{\text{reg}}^{-\frac{1}{2}}(D_{\text{reg}}-W)D_{\text{reg}}^{-\frac{1}{2}}.
\end{equation}

\begin{figure*}[th]
  \centering
  \includegraphics[width=1.8\columnwidth]{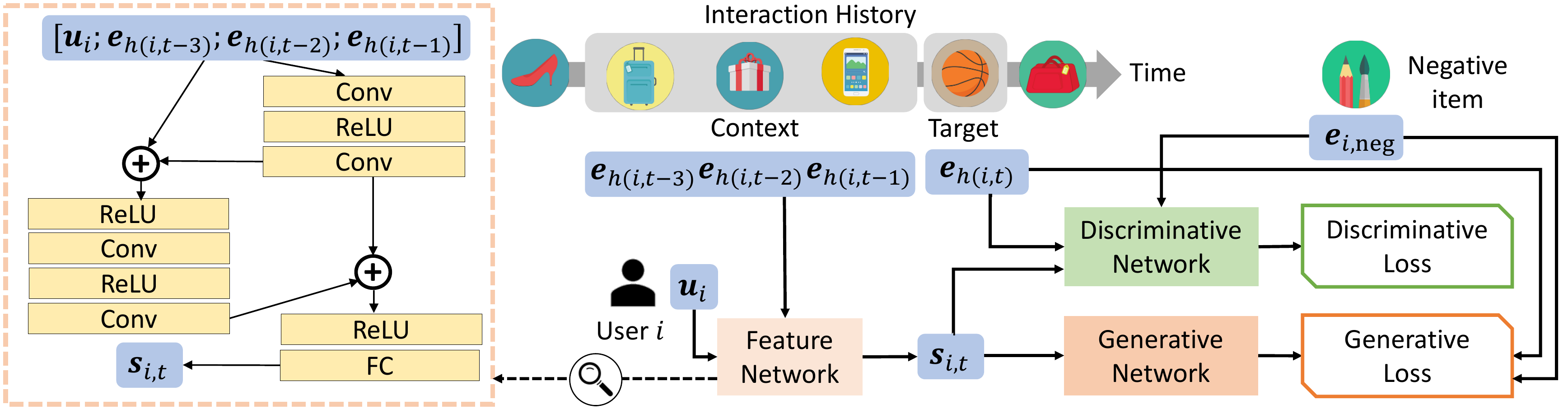} % Reduce the figure size so that it is slightly narrower than the column. Don't use precise values for figure width.This setup will avoid overfull boxes.
  \caption{The network architecture, including the Feature Network, which feeds into two branches: the Discriminative Network and the Generative Network. The Feature Network utilizes the residual structure in the box on the left.}
  \label{model}
\end{figure*}

To see the rationale behind the proposed technique, note that minimizing Eq. \ref{eq:le-optimization} will set the components $q_{i}$ and $q_{j}$ as close as possible. When the orthogonality constraints require some $q_{i}$ and $q_{j}$ to be different, the optimization makes them unequal only where the non-negative weight $W_{i,j}$ is small. When a node $i$ has small degree $d_i$, most of its $W_{i,j}$ are zero or small, so large changes in $q_{i}$ 
cause only negligible changes to the minimization objective. The additional regularization term penalizes large values for $q_{i}$, hence achieving regularization effects. 

The proposed regularization can also be understood from the probabilistic perspective \cite{Lawrence2012}. The eigenvector component $q_{i}$, which is associated with node $i$, may be seen as a random variable that we estimate from observations $W$. A small $d_i$ indicates that we have limited data for estimating $q_{i}$, resulting in a high variance in the estimates. Had we observed a few additional interactions for such nodes, their embeddings could change drastically. The proposed regularization involving $q_{i}^2$ can be understood as a Gaussian prior for $q_{i}$ with adaptive variance. %The variance of the prior increases monotonically with the node degree. 

\subsection{KNN Graph Construction}
To apply LE and \sysname, we construct two separate similarity graphs respectively for users and items. We will describe the construction of the user graph; item graph is constructed analogously.

Intuitively, two users are similar if they prefer the same items. Thus, we define the similarity between two users as the Jaccard index, or the intersection over union of the sets of items they interacted with \cite{knnjeccard}. Letting $\mathcal{I}_{i}$ denote the set of items that user $i$ interacted with, the similarity between user $i$ and user $j$ is
\begin{equation}
    \text{sim}(i,j) = \frac{|\mathcal{I}_{i} \cap \mathcal{I}_{j}|}{|\mathcal{I}_{i} \cup \mathcal{I}_{j}|},
\end{equation}
where $|\cdot|$ denotes the set cardinality. 
In the constructed graph $G$, we set $W_{i, j} = \text{sim}(i,j)$ if user $j$ is among the $K$ nearest neighbors of user $i$, \emph{or} if user $i$ is among the $K$ nearest neighbors of user $j$. Otherwise, $W_{i, j} = 0$.

\subsection{Computational Complexity}
% yinan: I moved the computational complexity part after KNN graph, I think it may be better:)
\label{sec:time-complexity}
Eigenvalue decomposition that solves for LE and \sysname{} embeddings has a high time complexity of $O(\mathcal{N}^3)$ for a graph with $\mathcal{N}$ nodes, which does not scale well to large datasets. 
Fortunately, for the $K$ nearest neighbor graph, the graph Laplacian is a sparse matrix with the sparsity larger than $1-\frac{2K^2}{\mathcal{N}^2}$. Utilizing this property, the Implicit Restarted Lanczos Method (IRLM) \cite{irlm,hou2018fast} computes the first $\mathcal{D}$ eigenvectors and eigenvalues with time complexity $O(K^2 \mathcal{D} \kappa)$, where $\kappa$ is the conditional number of $L$. Note that this complexity is independent of $\mathcal{N}$. 
In our experiment, it took only 172 seconds to compute 64 eigenvectors for the user KNN graph of the \textit{Anime} dataset, which has a $47143 \times 47143$ matrix and data sparsity of $96.85$\%. 
In comparison, it took 1888 seconds for the simple BPR model \cite{bpr} to converge.
The experiments were completed using two Intel Xeon 4114 CPUs. 
% \hl{compare this with regular eigendecomposition or BPR?}
% bpr time cost: 1888s
% regular eigendecomposition time cost: 

% When applying the LE and \sysname{} initialization to real-world data, an important consideration is the $O(n^3)$ time complexity of eigenvalue decomposition. 
% When dealing with large amounts of data, $O(n^3)$ can be prohibitively expensive. Fortunately, this problem has been extensively studied. A commonly used solution is the Nystr\"{o}m method \cite{Williams2001:SVD,Drineas05:SVD}, which samples $m$ ($m\ll n$) columns from $L$ and computes the first $\mathcal{D}$ eigenvectors and eigenvalues in $O(m^3+nm\mathcal{D})$ time. The method of \cite{LiMu2010:SVD} requires only $O(\mathcal{D}^3+nm\mathcal{D})$ time. 
% We would like to apply approximate eigendecomposition in \sysname{} in future work.

% Though the time complexity is $O(n^3)$, various effective implementations are proposed for sparse matrix eigendecomposition. For example, it only costs about 172 second in our largest dataset (user matrix of \textit{Anime}, size=$47143 \times 47143$, data sparsity=96.85\%) using scipy.sparse.linalg.svds on a server with two Intel(R) Xeon(R) Silver 4114 CPU @ 2.20GHz.

\section{Proposed Recommendation Model}
\label{sec:recmodel}
To maximize the benefits of the proposed initialization schemes, we design Dual-Loss Residual Recommendation (\modname{}).
Our preliminary experiments indicate that the immediate interaction history has the most predictive power for the next item that the user will interact with. For this reason, we select the last few items the user interacted with and the user embedding as the input features to the model. In the recommender system literature, this is known as sequential recommendation \cite{seq-wang19}. It is worth noting that LE and \sysname{} embeddings provide performance enhancement in a wide range of recommender systems, even though their effects are the most pronounced in sequential recommendation.  

% \modname{} prevents the degradation of representations by utilizing residual connections. We employ two different supervision signals to improve the recommendation performances.

% General recommendation systems tend to learn users’ static preferences by utilizing all historical interaction data with the assumption that all the data in the interaction sequence are of equal importance. However, this may not hold real-world scenarios since users’ choices towards items are not only based on their long-term preference but also short-term interest and the time-sensitive context (e.g., recently interacted items). Therefore, we propose to capture users' sequential behaviors in \modname{}.
% Therefore, we propose a sequential recommendation system, Dual-Loss Residual Recommendation (\modname), to make the full use of the well-initialized \sysname{} embeddings.

We denote the embeddings for user $i$ as $\bm{u}_i \in \mathbb{R}^{\mathcal{D}}$ and the embedding for item $j$ as $\bm{e}_j \in \mathbb{R}^{\mathcal{D}}$. After initialization, the embeddings are trained together with other network parameters. 
% The ground truth $y_{i, j, t} \in \{0, 1\}$ denotes if user $i$ interacted with item $j$ at time $t$. 
$h(i, t)$ denotes the item that user $i$ interacted with at time $t$. In order to represent the short-term interest of user $i$ at time $t$, we select the most recent $l$ items in the interaction history as $\bm{r}_{i, t} =  [\bm{e}_{h(i,t-l-1)}; \ldots; \bm{e}_{h(i,t-1)}]$. 

The network architecture is displayed in Figure~\ref{model}. The Feature Network takes $\bm{r}_{i, t}$ and $\bm{u}_i$ as input and outputs a common feature representation $\bm{s}_{i, t} \in \mathbb{R}^{\mathcal{D}}$. We employ residual connections \cite{resblock} in the Feature Network. Letting $x$ denote the input to a residual block, we feed $x$ into two convolution layers whose output is $F(x)$. The output of the block is $H(x) = F(x) + x$. We employ two residual blocks in the Feature Network, followed by a fully connected layer.

Residual connections may create invertible outputs \cite{behrmann2019invertible}, which preserve all the information in the input embeddings. We hypothesize that this may allow the network to fully utilize the power of initialization and put the hypothesis to test in Section \ref{sec:difarch}.

% We employ residual connections \cite{resblock} in the Feature Network, which is empirically proved to prevent degradation of representation quality towards CNN \cite{Alexander2019}. In other words, the information from the well-initialized embeddings of both users and items will be preserved to a large extent through residual connections. 

% \hl{add more descriptions}
% The Feature Network employs residual connections, first proposed by~\cite{resblock}, with convolutional layers and the ReLU activation function. 
% \hl{
% Letting $x$ be the input and the desired output is $H(x)$ in one residual block, we feed $x$ into two convolution layers and receive the current output $F(x)$. Then through the skip connection, we have $H(x) = F(x) + x$.
% }
% This network structure may be understood as an ensemble of networks of different depths \cite{larsson2017fractalnet}. 

% Most existing sequential recommendation methods directly predict scores for user-item pairs \cite{ncf,hgn}. Outside of the sequential recommendation literature, \cite{sparse-liu20} adopt a generative approach, which generates pseudo item embeddings that users may like. We believe both generative and discriminative aspects of the supervision work complementarily, thus we employ both supervision signals in \modname.

We feed the acquired feature $\bm{s}_{i, t}$ into both the Discriminative Network and the Generative Network, which exploit complementary aspects of the supervision.
The Discriminative Network $\mathcal{S}$ evaluates the preference of the user $i$ towards item $j$ at time $t$. We feed the user's current state $\bm{s}_{i, t}$ and the item embedding $\bm{e}_j$ into $\mathcal{S}$. The output $\mathcal{S}(\bm{s}_{i, t}, \bm{e}_j) \in \mathbb{R}$ represents the distance between users' preference and item $j$. More specifically, lower output value indicates higher similarity and better 
preference. 
We apply the square-square loss that suppresses the output distance of the correct item $\bm{e}_{h(i,t)}$ toward 0 and lifts the distance of a randomly sampled negative item $\bm{e}_{i, neg}$ above a margin $m_S$,
\begin{align}
\label{loss:discriminative}
\begin{split}
    \mathcal{L}_{S} & = \mathbb{E}\biggl[ \mathcal{S}(\bm{s}_{i,t}, \bm{e}_{h(i, t)})^2 + \left(\max( m_S - \mathcal{S}(\bm{s}_{i,t}, \bm{e}_{i, neg}), 0)\right)^2 \biggr],
\end{split}
\end{align}
% \begin{align}
% \label{loss:discriminative}
% \begin{split}
%     \mathcal{L}_{S} & = \mathbb{E}\biggl[ \mathcal{S}(\bm{s}_{i,t}, \bm{e}_{h(i, t)})^2 \biggr.\\
%     & + \biggl. \left(\max( m - \mathcal{S}(\bm{s}_{i,t}, \bm{e}_{i, neg}), 0)\right)^2 \biggr],
% \end{split}
% \end{align}
where $\mathbb{E}$ is the expectation over $i, t$ and negative samples. 

The Generative Network $\mathcal{G}$ directly predicts the embedding of items that the user prefers. Its output is $\mathcal{G}(\bm{s}_{i, t}) \in \mathbb{R}^{\mathcal{D}}$. We employ Euclidean distance to estimate the similarity between the network output $\mathcal{G}(\bm{s}_{i, t}) \in \mathbb{R}^{\mathcal{D}}$ and the evaluated item. 
The goal is to make the distance of the correct item $\bm{e}_{h(i,t)}$ smaller than the distance of the randomly sampled negative item $\bm{e}_{i, neg}$.
We adopt the following hinge loss with margin $m_{G}$,
\begin{align}
\label{loss:generative}
\begin{split}
     \mathcal{L}_{G} & = \mathbb{E}\biggl[ \max(\|\mathcal{G}(\bm{s}_{i, t}) - \bm{e}_{h(i, t)}\|_2 - \|\mathcal{G}(\bm{s}_{i, t}) - \bm{e}_{i, neg}\|_2 + m_{G}, 0)\biggr].
\end{split}
\end{align}
The final loss of the network is the sum of $\mathcal{L}_{G}$ and $\mathcal{L}_{S}$, which describe different aspects of supervisory signals. 

% \begin{align}
% \label{loss:generative}
% \begin{split}
%      \mathcal{L}_{G} & = \mathbb{E}\biggl[ \max(\|\mathcal{G}(\bm{s}_{i, t}) - \bm{e}_{h(i, t)}\|_2 \biggr. \\
%      & - \biggl. \|\mathcal{G}(\bm{s}_{i, t}) - \bm{e}_{i, neg}\|_2 + m, 0)\biggr].
% \end{split}
% \end{align}

%The experiments reveal that the two-branch architecture with the two losses provides substantial performance improvements. We also conduct ablation study to explore the effects of residual connections.

\section{Experiments}

\subsection{Datasets}
\begin{table}[t]
\caption{Dataset Statistics. Tail users and items are defined as the 25\% of users and items with the least interaction records. Their data shares are the proportions of their interactions out of all interaction records. }
\label{exp:datasets}
\resizebox{\columnwidth}{!}{
\begin{tabular}{@{}ccccccc@{}}
\toprule
Datasets    & \#Users & \#Items & \shortstack[t]{\# Inter-\\ actions} & \shortstack[t]{Data\\ Density}& \shortstack[t]{Tail Users'\\ Data Share}   & \shortstack[t]{Tail Items'\\ Data Share} \\ \midrule
\textit{ML-1M} & 6,040    & 3,650    & 1,000,127 & 4.54\%& 4.51\%  & 1.08\%\\
\textit{Steam} & 33,699   & 6,253    & 1,470,329 & 0.70\%& 11.67\%  & 1.83\%\\
\textit{Anime}        & 47,143   & 6,535    & 6,143,751 & 1.99\%& 5.73\% & 0.93\%\\ \bottomrule
\end{tabular}}
\end{table}
% \begin{table}[t]
% \caption{Dataset Statistics.}

% \label{exp:datasets}
% \resizebox{\columnwidth}{!}{
% \begin{tabular}{@{}ccccc@{}}
% \toprule
% Datasets    & \#Users & \#Items & Interactions& Data Density \\ \midrule
% \textit{ML-1M} & 6,040    & 3,650    & 1,000,127 & 4.5365\%\\
% \textit{Steam} & 33,699   & 6,253    & 1,470,329 & 0.6978\%\\
% \textit{Anime}        & 47,143   & 6,535    & 6,143,751 & 1.9942\%\\ \bottomrule
% \end{tabular}}
% \end{table}
The proposed algorithm is evaluated on three real-world datasets from different domains: \textit{MovieLens-1M} (\textit{ML-1M}) \cite{movielens}, \textit{Steam}\footnote{https://www.kaggle.com/trolukovich/steam-games-complete-dataset}, and \textit{Anime}\footnote{https://www.kaggle.com/CooperUnion/anime-recommendations-database}.
% \begin{itemize}
% \item \textit{MovieLens-1M} (\textit{ML-1M}) \cite{movielens}
% \item \textit{Steam} from Kaggle\footnote{https://www.kaggle.com/trolukovich/steam-games-complete-dataset}
% \item \textit{Anime} from Kaggle\footnote{https://www.kaggle.com/CooperUnion/anime-recommendations-database}
% \end{itemize}
Table \ref{exp:datasets} shows the dataset statistics. 
%The \textit{MovieLens-1M} dataset contains users' ratings of movies. Aiming to recommend items that a user would review favorably, we treat only ratings greater than three as positive ratings.  The \textit{Amazon-Books} dataset contains users' purchases. Since we aim to recommend items that a user would purchase, we regard all user-item interactions as positive. The \textit{Anime} dataset includes users' viewing history of anime movies and TV shows. 
Among the three datasets, \textit{Steam} is the most sparse with data sparsity over 99.3\%.
We split the datasets according to the time of interaction; the first 70\% interaction history of each user is used as the training set, the next 10\% as the validation set, and the rest as the testing set. 
% The experimental results are reliable using datasets with the well-defined train-test splits \cite{Steffen19}. 
For all three datasets, we retain only users and items with at least 20 interactions. When creating the candidate items for each user during testing, we use only items that the user has not interacted with in the training set or the validation set. 

\subsection{Initialization Methods}
We select a comprehensive list of initialization methods, including traditional initialization and state-of-the-art pre-training methods. 
\begin{itemize}
    \item Random embeddings drawn i.i.d. from a Gaussian distribution with zero mean and standard deviation 0.01.
    \item SVD, or singular value decomposition of the user-item interaction matrix $T \in \mathbb{R}^{N\times M}$. Let $T = U\Sigma V^\top$ and $\Tilde{\Sigma}$ denote the diagonal matrix with the largest $\mathcal{D}$ singular values. The user and item embeddings are $U\Tilde{\Sigma}^{\frac{1}{2}}$ and $\Tilde{\Sigma}^{\frac{1}{2}}V^\top$ respectively.
    \item BPR \cite{bpr}, 
    which ranks items based on the inner product of user and item embeddings.
    We take the trained embeddings as initialization for other models. 
    \item NetMF \cite{netmf}, which generalizes a variety of skip-gram models of graph embeddings such as DeepWalk \cite{deepwalk}, LINE \cite{line}, PTE \cite{pte} and node2vec \cite{node2vec}. NetMF empirically outperforms DeepWalk and LINE in conventional network mining tasks.
    % the learned KNN graph embedding of users and items from NetMF for a small window size (e.g., 1). It computes and factorizes DeepWalk matrix, and outperforms DeepWalk \cite{deepwalk} and LINE \cite{line} in conventional network mining tasks. \hl{relate to node2vec?}
    % \item NetMF\textsubscript{L} \cite{netmf}, NetMF for a large window size (e.g., 10), which is an approximation of NetMF\textsubscript{S}. \hl{why do we need an approximation if we already have S?}
    \item node2vec \cite{node2vec}. Following \cite{Zhu2020}, we trained node2vec on the heterogeneous graph with users and items as nodes and user-item interactions as edges. However, due to its high time complexity and high memory consumption \cite{pimentel2018fast}, we are unable to compute node2vec embeddings on \textit{Anime}.
    \item Graph-Bert \cite{zhang2020graphbert}, the user and item embeddings from Graph-Bert pretrained with self-supervised objectives. For fair comparisons with methods that do not use side information, we generate node attributes randomly.
\end{itemize}
For LE and \sysname{} embeddings, we employ the symmetrically normalized graph Laplacians ($L_{\text{sym}}$ in Eq. \ref{eq:l_sym} and $L_{\text{regsym}}$ in Eq. \ref{eq:l_regsym}).

\begin{table*}[t]
\small
\centering
\caption{Performances of different initialization for different neural recommendation methods. The best initialization results of each model are in bold faces and the second best results are underlined. ``Relative Imp.'' refers to the relative improvement of \sysname{} over the best baseline (other than LE), which is marked with $^*$.}
\label{exp:all1}
\begin{tabular}{llllllllllllll}

\toprule \multirow{2}{*}{Methods} &
\multirow{2}{*}{Initialization} & \multicolumn{4}{c}{Dataset: \textit{ML-1M}} & \multicolumn{4}{c}{Dataset: \textit{Steam}} & \multicolumn{4}{c}{Dataset: \textit{Anime}}\\
\cmidrule{3-14}

& & HR@5 & HR@10 & F1@5 & F1@10    & HR@5 & HR@10 & F1@5 & F1@10 & HR@5 & HR@10 & F1@5 & F1@10 \\ \toprule
\multirow{8}{*}{NCF} & random     & 0.0525          & 0.0874          & 0.0027          & 0.0044          & 0.0036          & 0.0090          & 0.0005          & 0.0009          & 0.0481          & 0.0591          & 0.0036          & 0.0037          \\
&SVD        & {\ul 0.2980}$^*$    & \textbf{0.4669}$^*$ & \textbf{0.0311}$^*$ & {\ul 0.0512}$^*$    & 0.0791          & 0.1557          & 0.0130          & 0.0203          & 0.0513          & 0.0892          & 0.0047          & 0.0081          \\
& BPR        & 0.2634          & 0.4096          & 0.0234          & 0.0388          & {\ul 0.0995}$^*$    & 0.1984          & {\ul 0.0166}$^*$    & \textbf{0.0269}$^*$ & 0.0490          & 0.0814          & 0.0044          & 0.0075          \\
% & NetMF\textsubscript{S}    & 0.2348          & 0.3826          & 0.0201          & 0.0352          & 0.0929          & 0.1885          & 0.0122          & 0.0251          & 0.0515          & 0.0872          & 0.0047          & 0.0079          \\
& NetMF    & 0.2151          & 0.3505          & 0.0180          & 0.0313          & 0.0912          & {\ul 0.2012}$^*$    & 0.0147          & 0.0267          & 0.0540          & 0.0910          & 0.0050$^*$          & 0.0084$^*$          \\
& node2vec & 0.2321 & 0.3776 & 0.0199 & 0.0340 & 0.0748 & 0.1475 & 0.0127 & 0.0196 &--  &--  &--& -- \\
& Graph-Bert & 0.1406          & 0.2354          & 0.0086          & 0.0154          & 0.0823          & 0.1577          & 0.0133          & 0.0196          & 0.0596$^*$          & {\ul 0.1044}$^*$    & 0.0041          & 0.0073          \\
& LE         & 0.2422          & 0.3916          & 0.0219          & 0.0369          & 0.0893          & 0.1766          & 0.0143          & 0.0229          & {\ul 0.0612}    & 0.1013          & {\ul 0.0051}    & {\ul 0.0088}    \\
& \sysname{}    & \textbf{0.3262} & {\ul 0.4515}    & {\ul 0.0311} & \textbf{0.0561} & \textbf{0.1027} & \textbf{0.2025} & \textbf{0.0167} & \textbf{0.0269} & \textbf{0.0678} & \textbf{0.1111} & \textbf{0.0054} & \textbf{0.0094}
 \\  \cmidrule{2-14}
\multicolumn{2}{r}{Relative Imp.} & 9.46\%  & -3.30\% & -0.02\% & 9.49\% & 3.22\% & 0.65\% & 0.70\% & -0.20\% & 13.76\% & 6.42\% & 7.71\% & 12.82\% \\
\bottomrule
\multirow{8}{*}{NGCF} & random     & 0.0725          & 0.1343          & 0.0039          & 0.0070          & 0.1072          & 0.1750          & 0.0179          & 0.0227          & 0.1789          & 0.3023          & 0.0253          & 0.0376          \\
& SVD        & 0.3164          & 0.4674          & 0.0360          & 0.0549          & 0.1050          & 0.1752          & 0.0175          & 0.0226          & 0.3585$^*$          & 0.5186$^*$          & 0.0461$^*$          & 0.0607$^*$         \\
& BPR        & 0.3301$^*$          & 0.4815$^*$          & 0.0420$^*$          & 0.0647$^*$          & 0.1071          & 0.1747          & 0.0179          & 0.0226          & 0.1719          & 0.2849          & 0.0338          & 0.0450          \\
% & NetMF\textsubscript{S}    & 0.1729          & 0.2677          & 0.0133          & 0.0244          & 0.0951          & 0.1604          & 0.0160          & 0.0204          & 0.1886          & 0.2961          & 0.0248          & 0.0377          \\
& NetMF    & 0.2873          & 0.4272          & 0.0297          & 0.0481          & 0.1016          & 0.1699          & 0.0173          & 0.0220          & 0.2355          & 0.3798          & 0.0342          & 0.0512          \\
& node2vec & 0.3195 & 0.4699 & 0.0369 & 0.0573 & {\ul 0.1193}$^*$ & {\ul 0.1939}$^*$ & {\ul 0.0205}$^*$ & {\ul 0.0255}$^*$ & -- & -- &--  & --\\
& Graph-Bert & 0.3169          & 0.4654          & 0.0354          & 0.0531          & 0.0980          & 0.1621          & 0.0161          & 0.0206          & 0.0847          & 0.1559          & 0.0059          & 0.0103          \\
& LE         & {\ul 0.3525}    & {\ul 0.5070}    & {\ul 0.0436}    & \textbf{0.0679} & {0.1102}    & {0.1818}    & {0.0184}    & {0.0238}    & {\ul 0.3739}    & {\ul 0.5270}    & {\ul 0.0496}    & {\ul 0.0760}    \\
& \sysname{}    & \textbf{0.3785} & \textbf{0.5258} & \textbf{0.0439} & {\ul 0.0667}    & \textbf{0.1342} & \textbf{0.2084} & \textbf{0.0228} & \textbf{0.0277} & \textbf{0.5098} & \textbf{0.6401} & \textbf{0.0712} & \textbf{0.1014}
 \\  \cmidrule{2-14}

 \multicolumn{2}{r}{Relative Imp.} & 14.66\% & 9.20\% & 4.66\% & 3.09\% & 12.49\% & 7.48\% & 11.21\% & 8.90\% & 42.20\% & 23.43\% & 54.56\% & 67.02\% \\
\bottomrule 
\multirow{8}{*}{DGCF} & random     & 0.3505          & 0.4964          & 0.0395          & 0.0611          & 0.1510          & 0.2321          & 0.0262          & 0.0317          & 0.2973          & 0.4474          & 0.0432          & 0.0618          \\
& SVD        &  0.3530    & 0.4982$^*$          & 0.0407          & 0.0635          & 0.1549          & 0.2365          & 0.0270          & 0.0325          & {\ul 0.3523}$^*$    & {\ul 0.5088}$^*$    & 0.0483$^*$          & 0.0697$^*$          \\
& BPR        & 0.3412          & 0.4964          & 0.0413          & 0.0647          & 0.1515          & 0.2371          & 0.0263          & 0.0320          & 0.2900          & 0.4403          & 0.0413          & 0.0600          \\
% & NetMF\textsubscript{S}    & 0.3440          & 0.4919          & 0.0419          & 0.0649          & 0.1545          & 0.2388          & 0.0271          & 0.0329          & 0.3002          & 0.4505          & 0.0422          & 0.0602          \\
& NetMF    & 0.3435          & 0.4977          & {\ul 0.0421}$^*$    & {\ul 0.0650}$^*$    & {\ul 0.1563}$^*$    & {\ul 0.2439}$^*$    & {\ul 0.0274}$^*$    & {\ul 0.0336}$^*$    & 0.3394          & 0.4963          & 0.0467          & 0.0677          \\
 & node2vec & {\ul 0.3533}$^*$ & 0.4926 & 0.0379 & 0.0597 & 0.1540 & 0.2420 & 0.0265 & 0.0325 & -- & -- & -- & -- \\
& Graph-Bert & 0.3483          & 0.4891          & 0.0366          & 0.0592          & 0.1487          & 0.2283          & 0.0256          & 0.0312          & 0.2939          & 0.4494          & 0.0420          & 0.0606          \\
& LE         & 0.3485          & {\ul 0.5032}    & 0.0408          & 0.0643          & 0.1551          & 0.2353          & 0.0269          & 0.0320          & 0.3506          & 0.5048          & {\ul 0.0488}    & {\ul 0.0702}    \\
& \sysname{}    & \textbf{0.3588} & \textbf{0.5045} & \textbf{0.0425} & \textbf{0.0653} & \textbf{0.1661} & \textbf{0.2530} & \textbf{0.0292} & \textbf{0.0349} & \textbf{0.5264} & \textbf{0.6477} & \textbf{0.0679} & \textbf{0.0916}

 \\  \cmidrule{2-14}

\multicolumn{2}{r}{Relative Imp.} & 1.56\% & 1.26\% & 0.97\% & 0.44\% & 6.27\% & 3.73\% & 6.46\% & 3.78\% & 49.42\% & 27.30\% & 40.77\% & 31.50\% \\
\bottomrule

\multirow{8}{*}{HGN} & random     & 0.3270          & 0.4879          & 0.0465          & 0.0688          & 0.1198          & 0.1977          & 0.0205          & 0.0261          & 0.1287          & 0.2457          & 0.0192          & 0.0305          \\
& SVD        & 0.3623          & 0.5315          & {\ul 0.0536}$^*$    & {\ul 0.0792}$^*$    & {\ul 0.1715}$^*$    & {0.2594}    & {\ul 0.0298}$^*$    & {0.0358}    & 0.4296$^*$          & 0.5989$^*$          & 0.0496          & 0.0764$^*$          \\
& BPR        & 0.3377          & 0.4985          & 0.0452          & 0.0688          & 0.1221          & 0.2058          & 0.0205          & 0.0272          & 0.2648          & 0.4070          & 0.0359          & 0.0529          \\
% & NetMF\textsubscript{S}    & 0.3510          & 0.5169          & 0.0464          & 0.0694          & 0.1342          & 0.2175          & 0.0226          & 0.0291          & 0.3895          & 0.5360          & 0.0484          & 0.0679          \\
& NetMF    & 0.3806$^*$          & 0.5425$^*$          & 0.0505          & 0.0737          & 0.1479          & 0.2338          & 0.0255          & 0.0318          & 0.4115          & 0.5552          & 0.0514$^*$          & 0.0704          \\
& node2vec & 0.3649 & 0.5333 & 0.0533 & 0.0783 & 0.1686 & {\ul 0.2694}$^*$ & 0.0297 & {\ul 0.0377}$^*$ & -- & -- & -- & -- \\
& Graph-Bert & 0.2934          & 0.4379          & 0.0324          & 0.0489          & 0.1304          & 0.2053          & 0.0221          & 0.0270          & 0.4147          & 0.5692          & 0.0506          & 0.0733          \\
& LE         & \textbf{0.3912} & {\ul 0.5460}    & 0.0507          & 0.0741          & 0.1689          & 0.2524          & 0.0293          & 0.0344          & {\ul 0.4753}    & {\ul 0.6374}    & {\ul 0.0605}    & {\ul 0.0875}    \\
& \sysname{}    & {\ul 0.3846}    & \textbf{0.5475} & \textbf{0.0543} & \textbf{0.0797} & \textbf{0.1822} & \textbf{0.2745} & \textbf{0.0318} & \textbf{0.0380} & \textbf{0.5010} & \textbf{0.6488} & \textbf{0.0622} & \textbf{0.0902}

\\  \cmidrule{2-14}

\multicolumn{2}{r}{Relative Imp.} & 1.05\% & 0.92\% & 1.41\% & 0.70\% & 6.24\% & 1.89\% & 6.75\% & 0.76\% & 16.62\% & 8.33\% & 21.08\% & 18.13\% \\
\bottomrule
\multirow{8}{*}{DLR$^2$} & random     & 0.3238          & 0.4538          & 0.0310          & 0.0497          & 0.1322          & 0.1912          & 0.0225          & 0.0247          & 0.4191          & 0.5976          & 0.0468          & 0.0744          \\
& SVD        & 0.4263          & 0.5820          & 0.0494          & 0.0775          & 0.1275          & 0.1913          & 0.0212          & 0.0241          & 0.5208          & 0.6469          & 0.0626          & 0.0831          \\
& BPR        & 0.4750          & 0.6273          & 0.0621          & 0.0979          & 0.1677 $^*$         & 0.2565$^*$          & 0.0290$^*$          & 0.0359$^*$          & 0.5822$^*$          & 0.7167$^*$          & 0.0909$^*$          & 0.1299$^*$          \\
% & NetMF\textsubscript{S}    & 0.4495          & 0.6060          & 0.0566          & 0.0902          & 0.1799$^*$          & 0.2642$^*$          & 0.0317$^*$          & 0.0378$^*$          & 0.4873          & 0.5990          & 0.0704          & 0.0973          \\
& NetMF    & 0.3969          & 0.5508          & 0.0487          & 0.0770          & 0.1522          & 0.2534          & 0.0260          & 0.0354          & 0.5526          & 0.6795          & 0.0820          & 0.1189          \\
& node2vec & {0.4947}$^*$ & {0.6603}$^*$ & {0.0640}$^*$ & {0.1011}$^*$ & 0.1504 & 0.2408 & 0.0256 & 0.0328 & -- & -- &--  &-- \\
& Graph-Bert & 0.3323          & 0.4553          & 0.0308          & 0.0497          & 0.1392          & 0.1976          & 0.0230          & 0.0251          & 0.4100          & 0.5708          & 0.0448          & 0.0700          \\ 

& LE         & {\ul 0.5119}    & {\ul 0.6730}    & {\ul 0.0692}    & {\ul 0.1070}    & {\ul 0.1853}    & {\ul 0.2878}    & {\ul 0.0332}    & {\ul 0.0420}    & {\ul 0.6065}    & {\ul 0.7438}    & {\ul 0.1010} & {\ul 0.1409}    \\ 
% Relative Imp. & 5.61\%  & 4.40\%  & 9.07\%  & 8.26\%  & 5.67\%  & 4.03\%  & 6.29\%  & 3.85\%  & 4.88\%  & 2.90\%  & 0.03\%  & 0.86\%  \\ \cmidrule{1-13}

& \sysname{}    & \textbf{0.5406} & \textbf{0.7026} & \textbf{0.0755} & \textbf{0.1158} & \textbf{0.1958} & \textbf{0.2994} & \textbf{0.0352} & \textbf{0.0437} & \textbf{0.6361} & \textbf{0.7654} & \textbf{0.1010} & \textbf{0.1421}

\\  \cmidrule{2-14}

\multicolumn{2}{r}{Relative Imp.} & 9.28\% & 6.41\% & 18.00\% & 14.62\% & 16.76\% & 16.73\% & 21.40\% & 21.68\% & 9.26\%  & 6.80\% & 11.08\% & 9.42\% \\ \bottomrule

\end{tabular}
\end{table*}

\subsection{Recommendation Baselines}

In order to estimate the effectiveness of our initialization methods, we select a comprehensive list of neural recommendation methods, including general collaborative filtering methods, NCF \cite{ncf}, graph convolution methods, NGCF \cite{NGCF19} and DGCF \cite{DGCF19}, and sequential recommendation methods, HGN \cite{hgn} and our proposed \modname{}. 

We also compare our recommendation model \modname{} to traditional recommendation methods used in \cite{recsys19}, including the non-personalized method, TopPop, which always recommends the most popular item. We also include traditional KNN approaches (UserKNN \cite{userknn} and ItemKNN \cite{itemknn}) and matrix factorization approaches (BPR \cite{bpr} and SLIM \cite{slim}).
Among the traditional baselines in \cite{recsys19}, we find P\textsuperscript{3}$\alpha$ and RP\textsuperscript{3}$\alpha$ to consistently underperform in all experiments and opt not to include them due to space limits. 
Experimental settings of these models can be found in the Appendix.

\begin{table*}[t]
\small
\centering
\caption{Performances of different algorithms. The best results of all are in bold faces and the second best results are underlined. The best results of the group above are marked with $^*$. Relative Imp. indicates to the relative improvement of \modname{} with \sysname{} initialization over the best baseline within the group.}
\label{exp:all2}
\begin{tabular}{lclllllllllllll}

\toprule
\multirow{2}{*}{Methods} & \multirow{2}{*}{Initialization} & \multicolumn{4}{c}{Dataset: \textit{ML-1M}} & \multicolumn{4}{c}{Dataset: \textit{Steam}} & \multicolumn{4}{c}{Dataset: \textit{Anime}}\\
\cmidrule{3-14}

&& HR@5 & HR@10 & F1@5 & F1@10    & HR@5 & HR@10 & F1@5 & F1@10 & HR@5 & HR@10 & F1@5 & F1@10 \\ \midrule
TopPop  &--& 0.3197          & 0.4377          & 0.0335          & 0.0517          & 0.1413          & 0.2306          & 0.0239          & 0.0307          & 0.4091          & 0.6103          & 0.0504          & 0.0807          \\
UserKNN &-- & 0.3851          & 0.5419$^*$          & 0.0483$^*$          & 0.0760 $^*$         & 0.1675          & 0.2575          & 0.0292          & 0.0357          & 0.4098          & 0.5899          & 0.0579          & 0.0802          \\
ItemKNN &-- & {0.3868}$^*$    & 0.5235          & 0.0462          & 0.0682          & 0.1729$^*$          & 0.2720$^*$          & 0.0309$^*$          & {0.0385}$^*$    & {0.5909}$^*$    & {0.7431}$^*$    & {0.0949}$^*$    & {0.1204}$^*$    \\
SLIM    &-- & 0.3692          & 0.5253          & 0.0435          & 0.0680          & 0.1671          & 0.2588          & 0.0290          & 0.0354          & 0.5337          & 0.7132          & 0.0684          & 0.1011          \\
BPR  & random   & 0.0851          & 0.1581          & 0.0107          & 0.0171          & 0.1582          & 0.2348          & 0.0279          & 0.0338          & 0.0406          & 0.0800          & 0.0074          & 0.0114          \\  \cmidrule{3-14}
\multicolumn{2}{c}{Relative Imp.}  & 39.76\% & 29.65\% & 56.18\% & 52.38\% & 13.24\% & 10.07\% & 13.98\% & 13.39\% & 7.65\% & 3.00\% & 6.43\% & 18.04\% \\
\midrule
BPR & \sysname{}     & 0.1094          & 0.2106          & 0.0151          & 0.0247          & 0.1638          & 0.2442          & 0.0293          & 0.0353          & 0.0810          & 0.1455          & 0.0170          & 0.0233          \\
NCF & \sysname{}     & 0.3262          & 0.4515          & 0.0311          & 0.0561          & 0.1027          & 0.2025          & 0.0167          & 0.0269          & 0.0678          & 0.1111          & 0.0054          & 0.0094          \\
NGCF & \sysname{}    & 0.3785          & 0.5258          & 0.0439          & 0.0667          & 0.1342          & 0.2084          & 0.0228          & 0.0277          & 0.5098          & 0.6401          & 0.0712$^*$          & 0.1014$^*$          \\
DGCF & \sysname{}    & 0.3588          & 0.5045          & 0.0425          & 0.0653          & 0.1661          & 0.2530          & 0.0292          & 0.0349          & 0.5264$^*$          & 0.6477          & 0.0679          & 0.0916          \\
HGN & \sysname{}     & 0.3846 $^*$         & {0.5475}$^*$    & {0.0543}$^*$    & {0.0797}$^*$    & {0.1822}$^*$    & {0.2745}$^*$    & {0.0318}$^*$    & 0.0380$^*$          & 0.5010          & 0.6488$^*$          & 0.0622          & 0.0902          \\ \cmidrule{3-14}
\multicolumn{2}{c}{Relative Imp.} & 40.56\% & 28.33\% & 38.98\% & 45.28\% & 7.46\%  & 9.07\%  & 10.93\% & 14.89\% & 20.84\% & 17.97\% & 41.84\% & 40.20\% 
\\ \midrule
\modname{} & LE        & {\ul 0.5119}    & {\ul 0.6730}    & {\ul 0.0692}    & {\ul 0.1070}    & {\ul 0.1853}    & {\ul 0.2878}    & {\ul 0.0332}    & {\ul 0.0420}    & {\ul 0.6065}    & {\ul 0.7438}    & {\ul 0.1010} & {\ul 0.1409}    \\ \cmidrule{3-14}
\multicolumn{2}{c}{Relative Imp.} & 5.61\%  & 4.40\%  & 9.07\%  & 8.26\%  & 5.67\%  & 4.03\%  & 6.29\%  & 3.85\%  & 4.88\%  & 2.90\%  & 0.03\%  & 0.86\%  \\ \midrule

\modname{} & \sysname{}   & \textbf{0.5406} & \textbf{0.7026} & \textbf{0.0755} & \textbf{0.1158} & \textbf{0.1958} & \textbf{0.2994} & \textbf{0.0352} & \textbf{0.0437} & \textbf{0.6361} & \textbf{0.7654} & \textbf{0.1010} & \textbf{0.1421}
\\ 
\bottomrule

\end{tabular}
\end{table*}
\subsection{Evaluation Metrics}
% In the KNN graphs, the number of nearest neighbors $K$ is set to 1000. The regularization coefficient $\alpha$ is set to $0.5$. The embedding size for all neural recommendation models is set to $64$. Other hyperparameters like learning rate and the number of training steps are tuned by grid search on the validation set. The initial learning rate is selected from the range $[0.01, 0.0001]$. More settings can be viewed in the Appendix.

To evaluate Top-N recommended items, we use hit ratio (HR@N) and F1 score (F1@N) at N = 1, 5, 10. Hit ratio intuitively measures whether the test item is present on the top-N list. F1 score is the harmonic mean of the precision and recall, which arguably gives a better measure of the incorrectly classified cases than the accuracy metric. We noticed the same tendency of precision and recall, and list the results only in terms of F1 score due to space limitation. 

\subsection{Results and Discussion}
\subsubsection{Comparison of Different Initialization}

Table \ref{exp:all1} shows the experimental results of different initialization on five different types of neural recommendation methods. We make the following observations. 
First, the choice of initialization has strong effects on the final performance. For example, for NCF on \textit{ML-1M} dataset, the F1 scores of \sysname{} are more than 10 times of random initialization. 
Second, \sysname{} consistently achieves the top position on most datasets and metrics; LE is a strong second, finishing in the top two in 32 out of the 60 metrics. For example, on  \textit{Anime} with the NGCF model, LE and \sysname{} lead the best baseline method by 25.22\% and 67.02\% respectively in F1@10. In DGCF and \modname{}, \sysname{} outperforms all other initialization methods. In other recommender systems, \sysname{} achieves the best performance in all but 1 or 2 metrics. 
Third, NGCF and DGCF benefit substantially from \sysname{}, even though they are GNN methods that already utilize graph information. This suggests \sysname{} is complementary to common graph networks. 

% These results demonstrate that the proposed initialization schemes provide significant boost to a wide range of neural recommendation systems. 

%, except F1@10 in NGCF method and HR@5 in HGN method on \textit{ML-1M} dataset. 
% According to Table \ref{exp:datasets}, the amount of interaction data for different users or items varies greatly in all datasets. Therefore, we can conclude that adding strong regularization to those tail users and items in LE help improve extracting representations for users and items. 
% We further do ablation study to see the performances on tail users.

\subsubsection{Comparison of Recommendation Models}
Table \ref{exp:all2} compares traditional recommendation methods, \sysname{} initialized recommendation methods, and \modname{} method with LE and \sysname{} initialization. Like \cite{recsys19}, we find well-tuned UserKNN and ItemKNN competitive. Among traditional recommendation methods, the two KNN methods claim the top spots on all metrics of all datasets. However, on \textit{ML-1M} and \textit{Steam}, \sysname-initialized neural recommendation baselines outperform the KNN methods. Together with the results of \cite{recsys19}, this finding highlights the importance of initialization for recommender systems. 

% Among neural baselines with \sysname{} initialization, the sequential recommendation method, HGN, appears to be the strongest and obtains the best performance on most metrics. 

Importantly, \modname{} with \sysname{} initialization outperforms all other methods, including strong traditional models and well-initialized neural recommendation models. For instance, on \textit{ML-1M} dataset and the F1@10 metric, our method leads traditional methods and neural baselines by 52.38\% and 45.28\%, respectively.

% Second, we observe \sysname{} initialized BPR outperforms randomly initialized one on all datasets, which illustrates that \sysname{} initialization improves traditional recommendation models as well.

% claims the top spots on all metrics on \textit{ML-1M} and \textit{Steam}. On \textit{Anime} dataset, the graph-based models, DGCF and NGCF, take the top spot for most metrics, except that HGN still leads on HR@10. These results demonstrate the advantage of considering users' sequential behaviors.

% Forth, comparing the performance of \modname{} with different initialization methods, \sysname{} outperforms all baselines, corroborating the findings in Table \ref{exp:all1} regarding the strengths of \sysname{}.

% \begin{figure*}%
%     \centering
%     \subfloat[\centering Performance of different architectures in Feature Network in \modname on \emph{Anime} dataset, initialized by \sysname{}. \label{fig:architecture}]{{\includegraphics[width=\columnwidth]{architecture_anime.png} }}%
%     \subfloat[\centering Performance of different networks for inferences with different loss signals. ``Both - G'' refers to using both $\mathcal{L}_{G}$ and $\mathcal{L}_{S}$ for training and using network $\mathcal{G}$ for inferences. \label{fig:dualloss}]{{\includegraphics[width=\columnwidth]{dualloss.png} }}
%     \caption{Ablation Study of \modname{}}
%     \label{fig:ablation}
% \end{figure*}

\begin{figure*}%
    \centering
    \subfloat[\centering Performance of different architectures in Feature Network in \modname{} on \emph{Anime} dataset, initialized by \sysname{}. \label{fig:architecture}]{{\includegraphics[width=2\columnwidth/3]{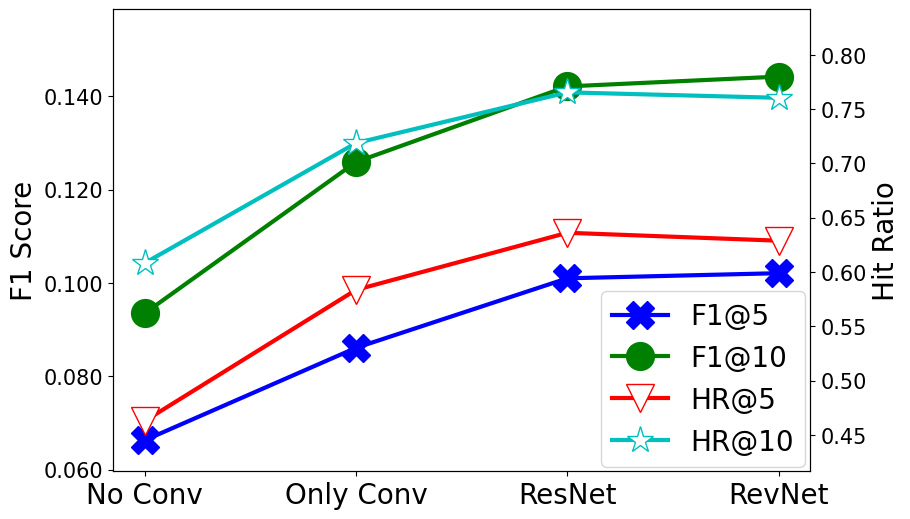} }}%
    % \subfloat[\centering Performance of different networks for inferences with different loss signals. ``Both - G'' refers to using both $\mathcal{L}_{G}$ and $\mathcal{L}_{S}$ for training and using network $\mathcal{G}$ for inferences. \label{fig:dualloss}]{{\includegraphics[width=2\columnwidth/4]{dualloss.png} }}
    \subfloat[\centering Performance of different regularization coefficient $\alpha$ in \sysname{} initialization on \emph{ML-1M} dataset, evaluated on \modname{}. \label{fig:alpha}]{{\includegraphics[width=2\columnwidth/3]{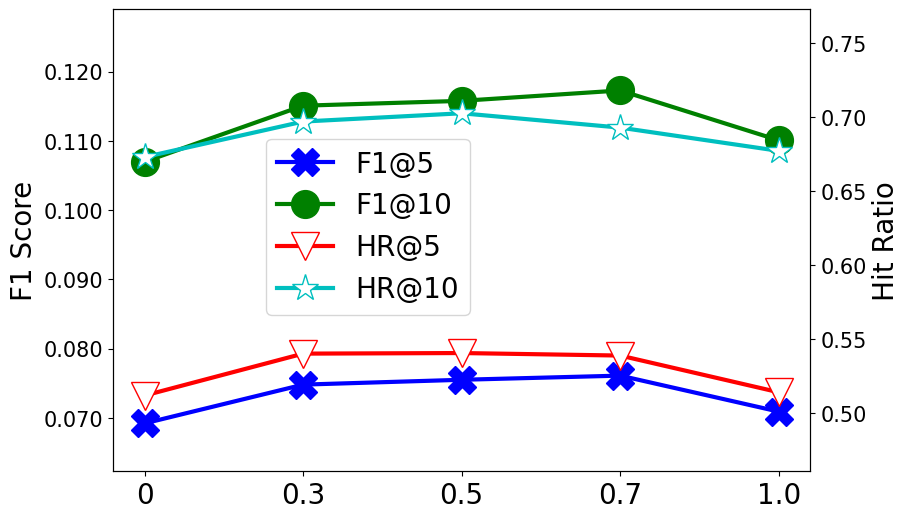} }}
    \subfloat[\centering Performance of different neighbor $K$ in \sysname{} initialization on \emph{ML-1M} dataset, evaluated on \modname. \label{fig:neighbor}]{{\includegraphics[width=2\columnwidth/3]{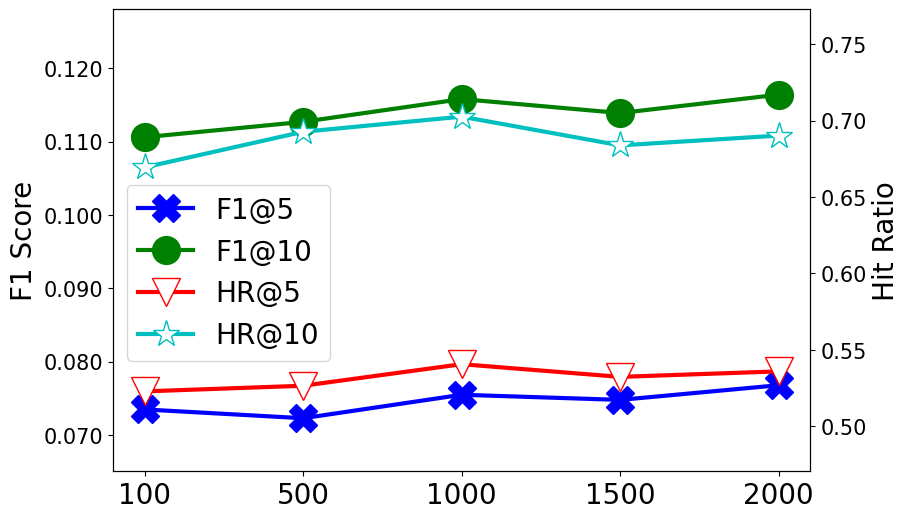} }}%
    \caption{Ablation Study}
    \label{fig:ablation}
\end{figure*}

% In this section, we first evaluate the performances of different initialization methods on tail users, to illustrate the strength of \sysname{} on the data-sparse regime. Then we compare the different architectures in the Feature Network in \modname to verify the advantage of residual connections. Finally we study the impact of hyperparameters, including the regularization coefficient $\alpha$ and the neighbor $K$ in in \sysname.

% \subsubsection{Performances on Tail Users}
\begin{table}[t]
\centering
\caption{Performances of different initialization methods on \modname{} on tail users with the least interaction data. ``Relative Imp.'' refers to the relative improvement of \sysname{} over the best baseline, including LE.}
\label{exp:bottom}
\begin{tabular}{cccccc}
\toprule
Dataset & Initialization & HR@5 & HR@10 & F1@5 & F1@10\\
\midrule

\multirow{5}{*}{\textit{ML-1M}} & random  & 0.1490          & 0.2490          & 0.0307          & 0.0382          \\
& SVD     & 0.2623          & 0.4166          & 0.0558          & 0.0722          \\
& LE      & 0.3192          & 0.4649          & 0.0724          & 0.0880          \\
& \sysname{} & \textbf{0.3536} & \textbf{0.5159} & \textbf{0.0823} & \textbf{0.1045} \\ \cmidrule{2-6}
& Relative Imp. & 10.78\%         & 10.97\%         & 13.67\%         & 18.72\%         \\
                                \cmidrule{1-6}
\multirow{5}{*}{\textit{Steam}}          & random       & 0.0817          & 0.1443          & 0.0177          & 0.0212          \\
 & SVD          & 0.0704          & 0.1072          & 0.0153          & 0.0155          \\
 & LE           & 0.0973          & 0.1684          & 0.0215          & 0.0258          \\
 & \sysname{}      & \textbf{0.1240} & \textbf{0.2030} & \textbf{0.0284} & \textbf{0.0321} \\ \cmidrule{2-6}
 & Relative Imp. & 27.44\%         & 20.55\%         & 31.99\%         & 24.44\%              \\ \cmidrule{1-6}
\multirow{5}{*}{\textit{Anime}}         & random       & 0.3209          & 0.4594          & 0.0680          & 0.0732          \\
 & SVD          & 0.2989          & 0.4661          & 0.0638          & 0.0765          \\
 & LE           & 0.3617          & 0.5205          & 0.0811          & 0.0955          \\
 & \sysname{}      & \textbf{0.3953} & \textbf{0.5408} & \textbf{0.0937} & \textbf{0.1195} \\\cmidrule{2-6}
 & Relative Imp. & 9.29\%          & 3.90\%          & 15.46\%         & 25.13\% \\
\bottomrule
\end{tabular}
\end{table}

\begin{table}[t]
\centering
\caption{Performances of networks trained with different supervision on \textit{ML-1M}.  
% ``Relative Imp.'' refers to the relative improvement of both loss signals over the single loss signal.
}
\vspace{-0.5cm}
\label{exp:dualloss}
\begin{tabular}{@{}ccllll@{}}
\\ \toprule
{\shortstack[t]{Inference\\ Branch} } & {\shortstack[t]{Loss\\ Function}}   & {HR@5} & {HR@10} & {F1@5} & {F1@10} \\ \midrule
\multirow{2}{*}{$\mathcal{G}$} & $\mathcal{L}_{G}$ & 0.3631    &   0.4982   & 0.0427 & 0.0601\\ 
 & $\mathcal{L}_{G} +\mathcal{L}_{S}$ & 0.4955    & 0.6510   & 0.0677 & 0.1012 \\\cmidrule{2-6}
\multicolumn{2}{r}{Relative Imp.} & 36.46\% & 30.67\% & 58.66\% & 68.41\% \\
\midrule
\multirow{2}{*}{$\mathcal{S}$} & $\mathcal{L}_{S}$ &  0.5228    & 0.6826   & 0.0744 & 0.1112   \\
 & $\mathcal{L}_{G} + \mathcal{L}_{S}$ & \textbf{0.5406}    & \textbf{0.7026}   & \textbf{0.0755} & \textbf{0.1158}  \\\cmidrule{2-6}
\multicolumn{2}{r}{Relative Imp.} & 3.40\% & 2.93\% & 1.49\% & 4.19\% \\
 \bottomrule
\end{tabular}
\end{table}

\subsubsection{Performance on Tail Users}
\sysname{} adaptively regularizes the embeddings of nodes with low degrees, which should benefit users with short interaction histories. In this experiment, we verify that claim by examining model performance on the 25\% users with the least interaction records. Selecting all of their interactions yields 9,771, 37,905, and 76,353 records for \emph{ML-1M}, \emph{Steam}, and \emph{Anime} test sets, respectively. We include all items because further exclusion would leave too few data for accurate performance measurement. 

We compare the performances of different initialization schemes on \modname{} in Table~\ref{exp:bottom}. We include SVD initialization because it is the best baseline on the most metrics (20 out of 60) from Table \ref{exp:all1}.
%\sysname{} outperforms LE embeddings on all experiments. 
\sysname{} achieves higher relative performance improvements on tail users than on the entire test sets, as shown in Table \ref{exp:all2}. 
%This is consistent with our expectation that the regularization of \sysname{} benefits users with limited interaction data. 
%We also observe that \sysname{} achieves the most significant improvements on \textit{Steam} dataset, which indicates that \sysname{} is more useful on sparser datasets.

% In particular, \sysname{} outperforms LE by over 50\% relatively in all metrics in \textit{Anime} datasets with the greatest improvements of 66.90\% and 62.83\% in HR@5 and F1@5.
% \hl{mention some relative improvements in Anime.} 
% However, on \emph{Amazon-Books}, the advantage of \sysname{} over LE is less prominent. This is consistent with our expectation as \emph{Amazon-Books} has the least tail data sparsity of the three datasets and hence requires the least regularization. 

\subsubsection{Effects of Architecture Design in \modname}
\label{sec:difarch}
The simple model of \modname{} has achieved the best performance thus far. We now explore architectural factors that contribute to its success. The literature suggests that residual connections may improve the quality of learned representation. In unsupervised learning, \cite{Alexander2019} finds that residual connections help to maintain the quality of learned representations in CNN throughout the network. \cite{behrmann2019invertible} proves that, under certain conditions, the output of residual networks contains all information needed to reconstruct the input. That is, the network is ``invertible''. This effect may facilitate \modname{} in fully utilizing the strength of the initialization schemes. 

We create three variations of \modname{} in order to observe the effects of residual connections and invertible representation. ``No Conv'' replaces all convolutional layers with fully connected layers and has no residual connections. ``Only Conv'' has convolutional layers but not the residual connections. ``ResNet'' denotes the complete \modname{} model. ``RevNet''  replaces the convolutional-residual design with the RevNet \cite{revnet} architecture, which is guaranteed to produce invertible features. Due to space limitation, we only present the results on \textit{Anime} dataset in Figure \ref{fig:architecture}. The results on other datasets are similar and can be found in the Appendix. 

Observing Figure \ref{fig:architecture}, we note that both convolutional layers and residual connections lead to significant performance improvements. The RevNet architecture achieves comparable performance with \modname{}, which suggests that learning invertible representations may facilitate the effective use of good initializations. 

% slightly modifies the design of the residual connections, which first splits the input $x$ channel-wise into two equal parts $x_1$ and $x_2$. The output is then the concatenation of $y_1 = x_1 + F(x_2)$ and $y2 = x2 + G(y_1)$. Both $F(\cdot)$ and $G(\cdot)$ includes two convolutional layers, activated by ReLU.

% , which inspired us to use residual connections in the \modname{}. In this section, we analyze the effectiveness of different architectures in the Feature Network in \modname{}. We found the results among the three datasets have similar trends. 

% As shown in Figure \ref{fig:architecture}, we can conclude that both convolutional layers and skip connections play significant roles in preserving the representation of input embeddings.

\subsubsection{Effects of the Two Loss Functions in \modname}
The \modname{} model employs a discriminative loss $\mathcal{L}_{S}$ and a generative loss $\mathcal{L}_{G}$. We now examine their effects on the final performance using an ablation study. Table \ref{exp:dualloss} shows the performance when one of the loss functions is removed. When adding $\mathcal{L}_{S}$ to the $\mathcal{L}_{G}$-only network, F1@10 increases relatively by 68.41\%. Similarly, adding the $\mathcal{L}_{G}$ branch to the $\mathcal{L}_{S}$-only network improves F1@10 relatively by 4.19\%.
The results demonstrate that we can harness significant synergy from the two-branch architecture with their respective loss functions.

\subsubsection{Sensitivity of Hyperparameters}
\label{sec:sensitivity}
% \begin{figure}[t]
%     \centering
%     \includegraphics[width=3.3in]{parameter.png}
%     \caption{Sensitivity of the embedding size $\mathcal{D}$ and the number of nearest neighbor $K$ on \textit{MovieLens-1M}.}
%     \label{fig:parameter}
% \end{figure}

We analyze the sensitivity of two hyperparameters: the population-based regularization coefficient $\alpha$ in Eq. \ref{eq:l_regsym} and the number of nearest neighbors $K$ used to construct the KNN graph. Due to space limitation, we only report the analysis on \textit{ML-1M} dataset on \modname{} model. 
Figure~\ref{fig:alpha} shows the results when $\alpha$ is selected from $\{0, 0.3, 0.5, 0.7, 1.0\}$. 
We observe that $\alpha=0.7$ achieves the best result.  
Figure \ref{fig:neighbor} shows the performances of \modname{} when the number of the nearest neighbors $K$ is varied. Though higher values of $K$ lead to small improvements, the sensitivity to this hyperparameter is low. 

\section{Conclusions}
In this paper, we propose a novel initialization scheme, \mbox{\sysname}, for neural recommendation systems. \mbox{\sysname} generates embeddings that capture neighborhood structures on the data manifold and adaptively regularize embeddings for tail users and items.
We also propose a sequential recommendation model, \modname{}, to make full use of the \sysname{}. 
Extensive experiments on three real-world datasets demonstrate the effectiveness of the proposed initialization strategy and network architecture. 
% \sysname{} displays superior performance in datasets with severe tail data sparsity, confirming the benefits of the adaptive regularization. 
We believe \sysname{} represents one step towards unleashing the full potential of deep neural networks in the area of recommendation systems. 

%\section{Acknowledgements}
\vspace{0.1in}
\noindent \textbf{Acknowledgments.} This research is supported by the National Research Foundation, Singapore under its AI Singapore  Programme (No. AISG-GC-2019-003), NRF Investigatorship (No. NRF-NRFI05-2019-0002), and NRF Fellowship (No. NRF-NRFF13-2021-0006), and by Alibaba Group through the Alibaba Innovative Research (AIR) Program and Alibaba-NTU Singapore Joint Research Institute (JRI).

% \section{Potential Ethical Impact}
% Recommender systems have broad applications in areas such as merchandise recommendation in online retail and destination recommendation in travel agencies. In certain applications like news or video recommendation, however, recommender systems may inadvertently create so-called ``echo chambers'' that shield users from diverse content. In other cases, they may recommend content that is considered extreme or harmful. We encourage machine learning practitioners to stay aware of the possible negative outcomes, evaluate the degree to which such effects exist, and take mitigative measures when necessary. 

% \begin{acks}
% TODO
% \end{acks}


\bibliographystyle{ACM-Reference-Format}
\bibliography{sample-base}


\begin{thebibliography}{61}

%%% the Web address.

\ifx \showCODEN    \undefined \def \showCODEN     #1{\unskip}     \fi
\ifx \showDOI      \undefined \def \showDOI       #1{#1}\fi
\ifx \showISBNx    \undefined \def \showISBNx     #1{\unskip}     \fi
\ifx \showISBNxiii \undefined \def \showISBNxiii  #1{\unskip}     \fi
\ifx \showISSN     \undefined \def \showISSN      #1{\unskip}     \fi
\ifx \showLCCN     \undefined \def \showLCCN      #1{\unskip}     \fi
\ifx \shownote     \undefined \def \shownote      #1{#1}          \fi
\ifx \showarticletitle \undefined \def \showarticletitle #1{#1}   \fi
\ifx \showURL      \undefined \def \showURL       {\relax}        \fi
% invisible to TeX
\providecommand\bibfield[2]{#2}
\providecommand\bibinfo[2]{#2}
\providecommand\natexlab[1]{#1}
\providecommand\showeprint[2][]{arXiv:#2}

\bibitem[\protect\citeauthoryear{Albert and Barabási}{Albert and
  Barabási}{2002}]%
        {BA-graph2002}
\bibfield{author}{\bibinfo{person}{Réka Albert} {and}
  \bibinfo{person}{Albert-László Barabási}.}
  \bibinfo{year}{2002}\natexlab{}.
\newblock \showarticletitle{Statistical mechanics of complex networks}.
\newblock \bibinfo{journal}{\emph{Reviews of Modern Physics}}
  (\bibinfo{year}{2002}).
\newblock


\bibitem[\protect\citeauthoryear{Allen-Zhu, Li, and Song}{Allen-Zhu
  et~al\mbox{.}}{2018}]%
        {allenzhu2018convergence}
\bibfield{author}{\bibinfo{person}{Zeyuan Allen-Zhu}, \bibinfo{person}{Yuanzhi
  Li}, {and} \bibinfo{person}{Zhao Song}.} \bibinfo{year}{2018}\natexlab{}.
\newblock \showarticletitle{A Convergence Theory for Deep Learning via
  Over-Parameterization}.
\newblock \bibinfo{journal}{\emph{arXiv 1811.03962}} (\bibinfo{year}{2018}).
\newblock


\bibitem[\protect\citeauthoryear{Behrmann, Grathwohl, Chen, Duvenaud, and
  Jacobsen}{Behrmann et~al\mbox{.}}{2019}]%
        {behrmann2019invertible}
\bibfield{author}{\bibinfo{person}{Jens Behrmann}, \bibinfo{person}{Will
  Grathwohl}, \bibinfo{person}{Ricky T.~Q. Chen}, \bibinfo{person}{David
  Duvenaud}, {and} \bibinfo{person}{Jörn-Henrik Jacobsen}.}
  \bibinfo{year}{2019}\natexlab{}.
\newblock \showarticletitle{Invertible Residual Networks}. In
  \bibinfo{booktitle}{\emph{ICML}}.
\newblock


\bibitem[\protect\citeauthoryear{{Belkin} and {Niyogi}}{{Belkin} and
  {Niyogi}}{2003}]%
        {initEigenmaps}
\bibfield{author}{\bibinfo{person}{M. {Belkin}} {and} \bibinfo{person}{P.
  {Niyogi}}.} \bibinfo{year}{2003}\natexlab{}.
\newblock \showarticletitle{Laplacian Eigenmaps for Dimensionality Reduction
  and Data Representation}.
\newblock \bibinfo{journal}{\emph{Neural Computation}} (\bibinfo{year}{2003}).
\newblock


\bibitem[\protect\citeauthoryear{Chang, Gao, He, Jin, and Li}{Chang
  et~al\mbox{.}}{2020}]%
        {graph-chang20}
\bibfield{author}{\bibinfo{person}{Jianxin Chang}, \bibinfo{person}{Chen Gao},
  \bibinfo{person}{Xiangnan He}, \bibinfo{person}{Depeng Jin}, {and}
  \bibinfo{person}{Yong Li}.} \bibinfo{year}{2020}\natexlab{}.
\newblock \showarticletitle{Bundle Recommendation with Graph Convolutional
  Networks}. In \bibinfo{booktitle}{\emph{SIGIR}}.
\newblock


\bibitem[\protect\citeauthoryear{Chen and Wong}{Chen and Wong}{2020}]%
        {graph-chen20}
\bibfield{author}{\bibinfo{person}{Tianwen Chen} {and} \bibinfo{person}{Raymond
  Chi-Wing Wong}.} \bibinfo{year}{2020}\natexlab{}.
\newblock \showarticletitle{Handling Information Loss of Graph Neural Networks
  for Session-Based Recommendation}. In \bibinfo{booktitle}{\emph{KDD}}.
\newblock


\bibitem[\protect\citeauthoryear{Chung}{Chung}{1997}]%
        {chung97}
\bibfield{author}{\bibinfo{person}{F.~R.~K. Chung}.}
  \bibinfo{year}{1997}\natexlab{}.
\newblock \bibinfo{booktitle}{\emph{Spectral Graph Theory}}.
\newblock \bibinfo{publisher}{American Mathematical Society}.
\newblock


\bibitem[\protect\citeauthoryear{Dacrema, Cremonesi, and Jannach}{Dacrema
  et~al\mbox{.}}{2019}]%
        {recsys19}
\bibfield{author}{\bibinfo{person}{Maurizio~Ferrari Dacrema},
  \bibinfo{person}{Paolo Cremonesi}, {and} \bibinfo{person}{Dietmar Jannach}.}
  \bibinfo{year}{2019}\natexlab{}.
\newblock \showarticletitle{Are We Really Making Much Progress? A Worrying
  Analysis of Recent Neural Recommendation Approaches}. In
  \bibinfo{booktitle}{\emph{RecSys}}.
\newblock


\bibitem[\protect\citeauthoryear{Devlin, Chang, Lee, and Toutanova}{Devlin
  et~al\mbox{.}}{2018}]%
        {devlin2018bert}
\bibfield{author}{\bibinfo{person}{Jacob Devlin}, \bibinfo{person}{Ming-Wei
  Chang}, \bibinfo{person}{Kenton Lee}, {and} \bibinfo{person}{Kristina
  Toutanova}.} \bibinfo{year}{2018}\natexlab{}.
\newblock \showarticletitle{BERT: Pre-training of Deep Bidirectional
  Transformers for Language Understanding}.
\newblock \bibinfo{journal}{\emph{arXiv 1810.04805}} (\bibinfo{year}{2018}).
\newblock


\bibitem[\protect\citeauthoryear{Du, Zhai, Poczos, and Singh}{Du
  et~al\mbox{.}}{2019}]%
        {du2019gradient}
\bibfield{author}{\bibinfo{person}{Simon~S. Du}, \bibinfo{person}{Xiyu Zhai},
  \bibinfo{person}{Barnabas Poczos}, {and} \bibinfo{person}{Aarti Singh}.}
  \bibinfo{year}{2019}\natexlab{}.
\newblock \showarticletitle{Gradient Descent Provably Optimizes
  Over-parameterized Neural Networks}. In \bibinfo{booktitle}{\emph{ICLR}}.
\newblock


\bibitem[\protect\citeauthoryear{Glorot and Bengio}{Glorot and Bengio}{2010}]%
        {glorot2010:init}
\bibfield{author}{\bibinfo{person}{Xavier Glorot} {and} \bibinfo{person}{Yoshua
  Bengio}.} \bibinfo{year}{2010}\natexlab{}.
\newblock \showarticletitle{Understanding the difficulty of training deep
  feedforward neural networks}. In \bibinfo{booktitle}{\emph{AISTATS}}.
\newblock


\bibitem[\protect\citeauthoryear{Gomez, Ren, Urtasun, and Grosse}{Gomez
  et~al\mbox{.}}{2017}]%
        {revnet}
\bibfield{author}{\bibinfo{person}{Aidan~N. Gomez}, \bibinfo{person}{Mengye
  Ren}, \bibinfo{person}{Raquel Urtasun}, {and} \bibinfo{person}{Roger~B.
  Grosse}.} \bibinfo{year}{2017}\natexlab{}.
\newblock \showarticletitle{The Reversible Residual Network: Backpropagation
  Without Storing Activations}.
\newblock \bibinfo{journal}{\emph{CoRR}}  \bibinfo{volume}{abs/1707.04585}
  (\bibinfo{year}{2017}).
\newblock


\bibitem[\protect\citeauthoryear{Grover and Leskovec}{Grover and
  Leskovec}{2016}]%
        {node2vec}
\bibfield{author}{\bibinfo{person}{Aditya Grover} {and} \bibinfo{person}{Jure
  Leskovec}.} \bibinfo{year}{2016}\natexlab{}.
\newblock \showarticletitle{node2vec: Scalable Feature Learning for Networks}.
\newblock \bibinfo{journal}{\emph{CoRR}}  \bibinfo{volume}{abs/1607.00653}
  (\bibinfo{year}{2016}).
\newblock


\bibitem[\protect\citeauthoryear{Harper and Konstan}{Harper and
  Konstan}{2015}]%
        {movielens}
\bibfield{author}{\bibinfo{person}{F.~Maxwell Harper} {and}
  \bibinfo{person}{Joseph~A. Konstan}.} \bibinfo{year}{2015}\natexlab{}.
\newblock \showarticletitle{The MovieLens Datasets: History and Context}.
\newblock \bibinfo{journal}{\emph{ACM Transactions on Interactive Intelligent
  Systems}} (\bibinfo{year}{2015}).
\newblock


\bibitem[\protect\citeauthoryear{He, Zhang, Ren, and Sun}{He
  et~al\mbox{.}}{2015a}]%
        {resblock}
\bibfield{author}{\bibinfo{person}{Kaiming He}, \bibinfo{person}{Xiangyu
  Zhang}, \bibinfo{person}{Shaoqing Ren}, {and} \bibinfo{person}{Jian Sun}.}
  \bibinfo{year}{2015}\natexlab{a}.
\newblock \showarticletitle{Deep Residual Learning for Image Recognition}.
\newblock \bibinfo{journal}{\emph{arXiv 1512.03385}} (\bibinfo{year}{2015}).
\newblock


\bibitem[\protect\citeauthoryear{He, Zhang, Ren, and Sun}{He
  et~al\mbox{.}}{2015b}]%
        {He_2015_ICCV}
\bibfield{author}{\bibinfo{person}{Kaiming He}, \bibinfo{person}{Xiangyu
  Zhang}, \bibinfo{person}{Shaoqing Ren}, {and} \bibinfo{person}{Jian Sun}.}
  \bibinfo{year}{2015}\natexlab{b}.
\newblock \showarticletitle{Delving Deep into Rectifiers: Surpassing
  Human-Level Performance on ImageNet Classification}. In
  \bibinfo{booktitle}{\emph{ICCV}}.
\newblock


\bibitem[\protect\citeauthoryear{He, Liao, Zhang, Nie, Hu, and Chua}{He
  et~al\mbox{.}}{2017}]%
        {ncf}
\bibfield{author}{\bibinfo{person}{Xiangnan He}, \bibinfo{person}{Lizi Liao},
  \bibinfo{person}{Hanwang Zhang}, \bibinfo{person}{Liqiang Nie},
  \bibinfo{person}{Xia Hu}, {and} \bibinfo{person}{Tat{-}Seng Chua}.}
  \bibinfo{year}{2017}\natexlab{}.
\newblock \showarticletitle{Neural Collaborative Filtering}.
\newblock \bibinfo{journal}{\emph{CoRR}}  \bibinfo{volume}{abs/1708.05031}
  (\bibinfo{year}{2017}).
\newblock


\bibitem[\protect\citeauthoryear{Hidasi, Karatzoglou, Baltrunas, and
  Tikk}{Hidasi et~al\mbox{.}}{2015}]%
        {seq-bal16}
\bibfield{author}{\bibinfo{person}{Balázs Hidasi}, \bibinfo{person}{Alexandros
  Karatzoglou}, \bibinfo{person}{Linas Baltrunas}, {and}
  \bibinfo{person}{Domonkos Tikk}.} \bibinfo{year}{2015}\natexlab{}.
\newblock \showarticletitle{Session-based Recommendations with Recurrent Neural
  Networks}.
\newblock  (\bibinfo{year}{2015}).
\newblock


\bibitem[\protect\citeauthoryear{Hou, Huang, Lam, and Zhang}{Hou
  et~al\mbox{.}}{2018}]%
        {hou2018fast}
\bibfield{author}{\bibinfo{person}{Thomas~Y. Hou}, \bibinfo{person}{De Huang},
  \bibinfo{person}{Ka~Chun Lam}, {and} \bibinfo{person}{Ziyun Zhang}.}
  \bibinfo{year}{2018}\natexlab{}.
\newblock \bibinfo{title}{A Fast Hierarchically Preconditioned Eigensolver
  Based On Multiresolution Matrix Decomposition}.
\newblock
\newblock
\showeprint[arxiv]{1804.03415}


\bibitem[\protect\citeauthoryear{Hu, Liu, Gomes, Zitnik, Liang, Pande, and
  Leskovec}{Hu et~al\mbox{.}}{2020b}]%
        {hu2020strategies}
\bibfield{author}{\bibinfo{person}{Weihua Hu}, \bibinfo{person}{Bowen Liu},
  \bibinfo{person}{Joseph Gomes}, \bibinfo{person}{Marinka Zitnik},
  \bibinfo{person}{Percy Liang}, \bibinfo{person}{Vijay Pande}, {and}
  \bibinfo{person}{Jure Leskovec}.} \bibinfo{year}{2020}\natexlab{b}.
\newblock \showarticletitle{Strategies for Pre-training Graph Neural Networks}.
  In \bibinfo{booktitle}{\emph{ICLR}}.
\newblock


\bibitem[\protect\citeauthoryear{Hu, Dong, Wang, Chang, and Sun}{Hu
  et~al\mbox{.}}{2020a}]%
        {gpt_gnn2020}
\bibfield{author}{\bibinfo{person}{Ziniu Hu}, \bibinfo{person}{Yuxiao Dong},
  \bibinfo{person}{Kuansan Wang}, \bibinfo{person}{Kai-Wei Chang}, {and}
  \bibinfo{person}{Yizhou Sun}.} \bibinfo{year}{2020}\natexlab{a}.
\newblock \showarticletitle{GPT-GNN: Generative Pre-Training of Graph Neural
  Networks}. In \bibinfo{booktitle}{\emph{KDD}}.
\newblock


\bibitem[\protect\citeauthoryear{Jin, Qin, Fang, Du, Zhang, Yu, Zhang, and
  Smola}{Jin et~al\mbox{.}}{2020}]%
        {jin2020efficient}
\bibfield{author}{\bibinfo{person}{Jiarui Jin}, \bibinfo{person}{Jiarui Qin},
  \bibinfo{person}{Yuchen Fang}, \bibinfo{person}{Kounianhua Du},
  \bibinfo{person}{Weinan Zhang}, \bibinfo{person}{Yong Yu},
  \bibinfo{person}{Zheng Zhang}, {and} \bibinfo{person}{Alexander~J. Smola}.}
  \bibinfo{year}{2020}\natexlab{}.
\newblock \bibinfo{title}{An Efficient Neighborhood-based Interaction Model for
  Recommendation on Heterogeneous Graph}.
\newblock
\newblock
\showeprint[arxiv]{2007.00216}


\bibitem[\protect\citeauthoryear{Kang and McAuley}{Kang and McAuley}{2018}]%
        {seq-wang18}
\bibfield{author}{\bibinfo{person}{Wang{-}Cheng Kang} {and}
  \bibinfo{person}{Julian~J. McAuley}.} \bibinfo{year}{2018}\natexlab{}.
\newblock \showarticletitle{Self-Attentive Sequential Recommendation}.
\newblock \bibinfo{journal}{\emph{arXiv 1808.09781}} (\bibinfo{year}{2018}).
\newblock


\bibitem[\protect\citeauthoryear{Kingma and Ba}{Kingma and Ba}{2015}]%
        {kingma2014adam}
\bibfield{author}{\bibinfo{person}{Diederik~P. Kingma} {and}
  \bibinfo{person}{Jimmy Ba}.} \bibinfo{year}{2015}\natexlab{}.
\newblock \showarticletitle{Adam: A Method for Stochastic Optimization}. In
  \bibinfo{booktitle}{\emph{ICLR}}.
\newblock


\bibitem[\protect\citeauthoryear{Kolesnikov, Zhai, and Beyer}{Kolesnikov
  et~al\mbox{.}}{2019}]%
        {Alexander2019}
\bibfield{author}{\bibinfo{person}{Alexander Kolesnikov},
  \bibinfo{person}{Xiaohua Zhai}, {and} \bibinfo{person}{Lucas Beyer}.}
  \bibinfo{year}{2019}\natexlab{}.
\newblock \showarticletitle{Revisiting Self-Supervised Visual Representation
  Learning}.
\newblock \bibinfo{journal}{\emph{CoRR}}  \bibinfo{volume}{abs/1901.09005}
  (\bibinfo{year}{2019}).
\newblock


\bibitem[\protect\citeauthoryear{Lawrence}{Lawrence}{2012}]%
        {Lawrence2012}
\bibfield{author}{\bibinfo{person}{Neil~D. Lawrence}.}
  \bibinfo{year}{2012}\natexlab{}.
\newblock \showarticletitle{A Unifying Probabilistic Perspective for Spectral
  Dimensionality Reduction: Insights and New Models}.
\newblock \bibinfo{journal}{\emph{Journal of Machine Learning Research}}
  (\bibinfo{year}{2012}).
\newblock


\bibitem[\protect\citeauthoryear{Lehoucq and Sorensen}{Lehoucq and
  Sorensen}{1996}]%
        {irlm}
\bibfield{author}{\bibinfo{person}{R.~B. Lehoucq} {and} \bibinfo{person}{D.~C.
  Sorensen}.} \bibinfo{year}{1996}\natexlab{}.
\newblock \showarticletitle{Deflation Techniques for an Implicitly Restarted
  Arnoldi Iteration}.
\newblock \bibinfo{journal}{\emph{SIAM J. Matrix Anal. Appl.}}
  (\bibinfo{year}{1996}).
\newblock


\bibitem[\protect\citeauthoryear{Liu, Zhang, Wu, Miao, Cui, Zhao, Zhao, and
  Guan}{Liu et~al\mbox{.}}{2019}]%
        {bpr2}
\bibfield{author}{\bibinfo{person}{Yong Liu}, \bibinfo{person}{Yinan Zhang},
  \bibinfo{person}{Qiong Wu}, \bibinfo{person}{Chunyan Miao},
  \bibinfo{person}{Lizhen Cui}, \bibinfo{person}{Binqiang Zhao},
  \bibinfo{person}{Yin Zhao}, {and} \bibinfo{person}{Lu Guan}.}
  \bibinfo{year}{2019}\natexlab{}.
\newblock \showarticletitle{Diversity-Promoting Deep Reinforcement Learning for
  Interactive Recommendation}.
\newblock \bibinfo{journal}{\emph{CoRR}}  \bibinfo{volume}{abs/1903.07826}
  (\bibinfo{year}{2019}).
\newblock


\bibitem[\protect\citeauthoryear{{Liu}, {Zheng}, {Zhang}, {Han}, and
  {Yu}}{{Liu} et~al\mbox{.}}{2019}]%
        {spectral-Liu2019}
\bibfield{author}{\bibinfo{person}{Z. {Liu}}, \bibinfo{person}{L. {Zheng}},
  \bibinfo{person}{J. {Zhang}}, \bibinfo{person}{J. {Han}}, {and}
  \bibinfo{person}{P.~S. {Yu}}.} \bibinfo{year}{2019}\natexlab{}.
\newblock \showarticletitle{JSCN: Joint Spectral Convolutional Network for
  Cross Domain Recommendation}. In \bibinfo{booktitle}{\emph{IEEE Big Data}}.
\newblock


\bibitem[\protect\citeauthoryear{Ma, Kang, and Liu}{Ma et~al\mbox{.}}{2019}]%
        {hgn}
\bibfield{author}{\bibinfo{person}{Chen Ma}, \bibinfo{person}{Peng Kang}, {and}
  \bibinfo{person}{Xue Liu}.} \bibinfo{year}{2019}\natexlab{}.
\newblock \showarticletitle{Hierarchical Gating Networks for Sequential
  Recommendation}. In \bibinfo{booktitle}{\emph{KDD}}.
\newblock


\bibitem[\protect\citeauthoryear{Ng, Jordan, and Weiss}{Ng
  et~al\mbox{.}}{2001}]%
        {andrew02}
\bibfield{author}{\bibinfo{person}{Andrew~Y. Ng}, \bibinfo{person}{Michael~I.
  Jordan}, {and} \bibinfo{person}{Yair Weiss}.}
  \bibinfo{year}{2001}\natexlab{}.
\newblock \showarticletitle{On Spectral Clustering: Analysis and an Algorithm}.
  In \bibinfo{booktitle}{\emph{NeurIPS}}.
\newblock


\bibitem[\protect\citeauthoryear{Ning and Karypis}{Ning and Karypis}{2011}]%
        {slim}
\bibfield{author}{\bibinfo{person}{Xia Ning} {and} \bibinfo{person}{George
  Karypis}.} \bibinfo{year}{2011}\natexlab{}.
\newblock \showarticletitle{SLIM: Sparse Linear Methods for Top-N Recommender
  Systems}.
\newblock \bibinfo{journal}{\emph{ICDM}} (\bibinfo{year}{2011}).
\newblock


\bibitem[\protect\citeauthoryear{Perozzi, Al-Rfou, and Skiena}{Perozzi
  et~al\mbox{.}}{2014}]%
        {deepwalk}
\bibfield{author}{\bibinfo{person}{Bryan Perozzi}, \bibinfo{person}{Rami
  Al-Rfou}, {and} \bibinfo{person}{Steven Skiena}.}
  \bibinfo{year}{2014}\natexlab{}.
\newblock \showarticletitle{DeepWalk: Online Learning of Social
  Representations}. In \bibinfo{booktitle}{\emph{KDD}}.
\newblock


\bibitem[\protect\citeauthoryear{Pimentel, Veloso, and Ziviani}{Pimentel
  et~al\mbox{.}}{2018}]%
        {pimentel2018fast}
\bibfield{author}{\bibinfo{person}{Tiago Pimentel}, \bibinfo{person}{Adriano
  Veloso}, {and} \bibinfo{person}{Nivio Ziviani}.}
  \bibinfo{year}{2018}\natexlab{}.
\newblock \showarticletitle{Fast Node Embeddings: Learning Ego-Centric
  Representations}.
\newblock \bibinfo{journal}{\emph{ICLR (Workshop)}}.
\newblock


\bibitem[\protect\citeauthoryear{Qin, Ren, Fang, Zhang, and Yu}{Qin
  et~al\mbox{.}}{2020}]%
        {graph-qin20}
\bibfield{author}{\bibinfo{person}{Jiarui Qin}, \bibinfo{person}{Kan Ren},
  \bibinfo{person}{Yuchen Fang}, \bibinfo{person}{Weinan Zhang}, {and}
  \bibinfo{person}{Yong Yu}.} \bibinfo{year}{2020}\natexlab{}.
\newblock \showarticletitle{Sequential Recommendation with Dual Side
  Neighbor-Based Collaborative Relation Modeling}. In
  \bibinfo{booktitle}{\emph{WSDM}}.
\newblock


\bibitem[\protect\citeauthoryear{Qiu, Dong, Ma, Li, Wang, and Tang}{Qiu
  et~al\mbox{.}}{2017}]%
        {netmf}
\bibfield{author}{\bibinfo{person}{Jiezhong Qiu}, \bibinfo{person}{Yuxiao
  Dong}, \bibinfo{person}{Hao Ma}, \bibinfo{person}{Jian Li},
  \bibinfo{person}{Kuansan Wang}, {and} \bibinfo{person}{Jie Tang}.}
  \bibinfo{year}{2017}\natexlab{}.
\newblock \showarticletitle{Network Embedding as Matrix Factorization: Unifying
  DeepWalk, LINE, PTE, and node2vec}.
\newblock \bibinfo{journal}{\emph{CoRR}}  \bibinfo{volume}{abs/1710.02971}
  (\bibinfo{year}{2017}).
\newblock


\bibitem[\protect\citeauthoryear{Qu, Zhu, Duan, and Shi}{Qu
  et~al\mbox{.}}{2020}]%
        {graph-qu20}
\bibfield{author}{\bibinfo{person}{Liang Qu}, \bibinfo{person}{Huaisheng Zhu},
  \bibinfo{person}{Qiqi Duan}, {and} \bibinfo{person}{Yuhui Shi}.}
  \bibinfo{year}{2020}\natexlab{}.
\newblock \showarticletitle{Continuous-Time Link Prediction via Temporal
  Dependent Graph Neural Network}. In \bibinfo{booktitle}{\emph{WWW}}.
\newblock


\bibitem[\protect\citeauthoryear{Rendle, Freudenthaler, Gantner, and
  Schmidt-Thieme}{Rendle et~al\mbox{.}}{2009}]%
        {bpr}
\bibfield{author}{\bibinfo{person}{Steffen Rendle}, \bibinfo{person}{Christoph
  Freudenthaler}, \bibinfo{person}{Zeno Gantner}, {and} \bibinfo{person}{Lars
  Schmidt-Thieme}.} \bibinfo{year}{2009}\natexlab{}.
\newblock \showarticletitle{BPR: Bayesian Personalized Ranking from Implicit
  Feedback}. In \bibinfo{booktitle}{\emph{UAI}}.
\newblock


\bibitem[\protect\citeauthoryear{Rendle, Zhang, and Koren}{Rendle
  et~al\mbox{.}}{2019}]%
        {Steffen19}
\bibfield{author}{\bibinfo{person}{Steffen Rendle}, \bibinfo{person}{Li Zhang},
  {and} \bibinfo{person}{Yehuda Koren}.} \bibinfo{year}{2019}\natexlab{}.
\newblock \showarticletitle{On the Difficulty of Evaluating Baselines: {A}
  Study on Recommender Systems}.
\newblock \bibinfo{journal}{\emph{CoRR}}  \bibinfo{volume}{abs/1905.01395}
  (\bibinfo{year}{2019}).
\newblock


\bibitem[\protect\citeauthoryear{Sarwar, Karypis, Konstan, and Riedl}{Sarwar
  et~al\mbox{.}}{2001}]%
        {userknn}
\bibfield{author}{\bibinfo{person}{Badrul Sarwar}, \bibinfo{person}{George
  Karypis}, \bibinfo{person}{Joseph Konstan}, {and} \bibinfo{person}{John
  Riedl}.} \bibinfo{year}{2001}\natexlab{}.
\newblock \showarticletitle{Item-Based Collaborative Filtering Recommendation
  Algorithms}. In \bibinfo{booktitle}{\emph{WWW}}.
\newblock


\bibitem[\protect\citeauthoryear{Song, Xiao, Wang, Charlin, Zhang, and
  Tang}{Song et~al\mbox{.}}{2019}]%
        {graph-song19}
\bibfield{author}{\bibinfo{person}{Weiping Song}, \bibinfo{person}{Zhiping
  Xiao}, \bibinfo{person}{Yifan Wang}, \bibinfo{person}{Laurent Charlin},
  \bibinfo{person}{Ming Zhang}, {and} \bibinfo{person}{Jian Tang}.}
  \bibinfo{year}{2019}\natexlab{}.
\newblock \showarticletitle{Session-Based Social Recommendation via Dynamic
  Graph Attention Networks}. In \bibinfo{booktitle}{\emph{WSDM}}.
\newblock


\bibitem[\protect\citeauthoryear{Sun, Zhang, Dong, Zhang, Li, Zhang, and
  Min}{Sun et~al\mbox{.}}{2017}]%
        {knnjeccard}
\bibfield{author}{\bibinfo{person}{Shuangbo Sun}, \bibinfo{person}{Zhiheng
  Zhang}, \bibinfo{person}{Xinling Dong}, \bibinfo{person}{Hengru Zhang},
  \bibinfo{person}{Tongjun Li}, \bibinfo{person}{Lin Zhang}, {and}
  \bibinfo{person}{Fan Min}.} \bibinfo{year}{2017}\natexlab{}.
\newblock \showarticletitle{Integrating Triangle and Jaccard similarities for
  recommendation}.
\newblock \bibinfo{journal}{\emph{PLoS ONE}} (\bibinfo{year}{2017}).
\newblock


\bibitem[\protect\citeauthoryear{Sutskever, Martens, Dahl, and
  Hinton}{Sutskever et~al\mbox{.}}{2013}]%
        {Sutskever2013:important_init}
\bibfield{author}{\bibinfo{person}{Ilya Sutskever}, \bibinfo{person}{James
  Martens}, \bibinfo{person}{George Dahl}, {and} \bibinfo{person}{Geoffrey
  Hinton}.} \bibinfo{year}{2013}\natexlab{}.
\newblock \showarticletitle{On the Importance of Initialization and Momentum in
  Deep Learning}. In \bibinfo{booktitle}{\emph{ICML}}.
\newblock


\bibitem[\protect\citeauthoryear{Tang, Qu, and Mei}{Tang
  et~al\mbox{.}}{2015a}]%
        {pte}
\bibfield{author}{\bibinfo{person}{Jian Tang}, \bibinfo{person}{Meng Qu}, {and}
  \bibinfo{person}{Qiaozhu Mei}.} \bibinfo{year}{2015}\natexlab{a}.
\newblock \showarticletitle{{PTE:} Predictive Text Embedding through
  Large-scale Heterogeneous Text Networks}.
\newblock \bibinfo{journal}{\emph{CoRR}}  \bibinfo{volume}{abs/1508.00200}
  (\bibinfo{year}{2015}).
\newblock


\bibitem[\protect\citeauthoryear{Tang, Qu, Wang, Zhang, Yan, and Mei}{Tang
  et~al\mbox{.}}{2015b}]%
        {line}
\bibfield{author}{\bibinfo{person}{Jian Tang}, \bibinfo{person}{Meng Qu},
  \bibinfo{person}{Mingzhe Wang}, \bibinfo{person}{Ming Zhang},
  \bibinfo{person}{Jun Yan}, {and} \bibinfo{person}{Qiaozhu Mei}.}
  \bibinfo{year}{2015}\natexlab{b}.
\newblock \showarticletitle{{LINE:} Large-scale Information Network Embedding}.
\newblock \bibinfo{journal}{\emph{CoRR}}  \bibinfo{volume}{abs/1503.03578}
  (\bibinfo{year}{2015}).
\newblock


\bibitem[\protect\citeauthoryear{Tang and Wang}{Tang and Wang}{2018}]%
        {caser}
\bibfield{author}{\bibinfo{person}{Jiaxi Tang} {and} \bibinfo{person}{Ke
  Wang}.} \bibinfo{year}{2018}\natexlab{}.
\newblock \showarticletitle{Personalized Top-N Sequential Recommendation via
  Convolutional Sequence Embedding}. In \bibinfo{booktitle}{\emph{WSDM}}.
\newblock


\bibitem[\protect\citeauthoryear{Wang, de~Vries, and Reinders}{Wang
  et~al\mbox{.}}{2006}]%
        {itemknn}
\bibfield{author}{\bibinfo{person}{Jun Wang}, \bibinfo{person}{Arjen~P. de
  Vries}, {and} \bibinfo{person}{Marcel J.~T. Reinders}.}
  \bibinfo{year}{2006}\natexlab{}.
\newblock \showarticletitle{Unifying User-Based and Item-Based Collaborative
  Filtering Approaches by Similarity Fusion}. In
  \bibinfo{booktitle}{\emph{SIGIR}}.
\newblock


\bibitem[\protect\citeauthoryear{Wang, Hu, Wang, Cao, Sheng, and Orgun}{Wang
  et~al\mbox{.}}{2019b}]%
        {seq-wang19}
\bibfield{author}{\bibinfo{person}{Shoujin Wang}, \bibinfo{person}{Liang Hu},
  \bibinfo{person}{Yan Wang}, \bibinfo{person}{Longbing Cao},
  \bibinfo{person}{Quan~Z. Sheng}, {and} \bibinfo{person}{Mehmet Orgun}.}
  \bibinfo{year}{2019}\natexlab{b}.
\newblock \showarticletitle{Sequential Recommender Systems: Challenges,
  Progress and Prospects}. In \bibinfo{booktitle}{\emph{IJCAI}}.
\newblock


\bibitem[\protect\citeauthoryear{Wang, He, Wang, Feng, and Chua}{Wang
  et~al\mbox{.}}{2019a}]%
        {NGCF19}
\bibfield{author}{\bibinfo{person}{Xiang Wang}, \bibinfo{person}{Xiangnan He},
  \bibinfo{person}{Meng Wang}, \bibinfo{person}{Fuli Feng}, {and}
  \bibinfo{person}{Tat{-}Seng Chua}.} \bibinfo{year}{2019}\natexlab{a}.
\newblock \showarticletitle{Neural Graph Collaborative Filtering}. In
  \bibinfo{booktitle}{\emph{SIGIR}}.
\newblock


\bibitem[\protect\citeauthoryear{Wang, Jin, Zhang, He, Xu, and Chua}{Wang
  et~al\mbox{.}}{2020a}]%
        {DGCF19}
\bibfield{author}{\bibinfo{person}{Xiang Wang}, \bibinfo{person}{Hongye Jin},
  \bibinfo{person}{An Zhang}, \bibinfo{person}{Xiangnan He},
  \bibinfo{person}{Tong Xu}, {and} \bibinfo{person}{Tat{-}Seng Chua}.}
  \bibinfo{year}{2020}\natexlab{a}.
\newblock \showarticletitle{Disentangled Graph Collaborative Filtering}. In
  \bibinfo{booktitle}{\emph{SIGIR}}.
\newblock


\bibitem[\protect\citeauthoryear{Wang, Wei, Cong, Li, Mao, and Qiu}{Wang
  et~al\mbox{.}}{2020b}]%
        {graph-wang20}
\bibfield{author}{\bibinfo{person}{Ziyang Wang}, \bibinfo{person}{Wei Wei},
  \bibinfo{person}{Gao Cong}, \bibinfo{person}{Xiao-Li Li},
  \bibinfo{person}{Xian-Ling Mao}, {and} \bibinfo{person}{Minghui Qiu}.}
  \bibinfo{year}{2020}\natexlab{b}.
\newblock \showarticletitle{Global Context Enhanced Graph Neural Networks for
  Session-Based Recommendation}. In \bibinfo{booktitle}{\emph{SIGIR}}.
\newblock


\bibitem[\protect\citeauthoryear{Wu, Ahmed, Beutel, Smola, and Jing}{Wu
  et~al\mbox{.}}{2017}]%
        {seq-wu17}
\bibfield{author}{\bibinfo{person}{Chao-Yuan Wu}, \bibinfo{person}{Amr Ahmed},
  \bibinfo{person}{Alex Beutel}, \bibinfo{person}{Alexander~J. Smola}, {and}
  \bibinfo{person}{How Jing}.} \bibinfo{year}{2017}\natexlab{}.
\newblock \showarticletitle{Recurrent Recommender Networks}. In
  \bibinfo{booktitle}{\emph{WSDM}}.
\newblock


\bibitem[\protect\citeauthoryear{Wu, Tang, Zhu, Wang, Xie, and Tan}{Wu
  et~al\mbox{.}}{2018}]%
        {seq-wu18}
\bibfield{author}{\bibinfo{person}{Shu Wu}, \bibinfo{person}{Yuyuan Tang},
  \bibinfo{person}{Yanqiao Zhu}, \bibinfo{person}{Liang Wang},
  \bibinfo{person}{Xing Xie}, {and} \bibinfo{person}{Tieniu Tan}.}
  \bibinfo{year}{2018}\natexlab{}.
\newblock \showarticletitle{Session-based Recommendation with Graph Neural
  Networks}.
\newblock \bibinfo{journal}{\emph{arXiv 1811.00855}}.
\newblock


\bibitem[\protect\citeauthoryear{Wu, Zhang, Sun, and Cui}{Wu
  et~al\mbox{.}}{2020}]%
        {wu2020graph}
\bibfield{author}{\bibinfo{person}{Shiwen Wu}, \bibinfo{person}{Wentao Zhang},
  \bibinfo{person}{Fei Sun}, {and} \bibinfo{person}{Bin Cui}.}
  \bibinfo{year}{2020}\natexlab{}.
\newblock \bibinfo{title}{Graph Neural Networks in Recommender Systems: A
  Survey}.
\newblock
\newblock
\showeprint[arxiv]{2011.02260}


\bibitem[\protect\citeauthoryear{Xu, Zhao, Liu, Sheng, Xu, Zhuang, Fang, and
  Zhou}{Xu et~al\mbox{.}}{2019}]%
        {seq-xu19}
\bibfield{author}{\bibinfo{person}{Chengfeng Xu}, \bibinfo{person}{Pengpeng
  Zhao}, \bibinfo{person}{Yanchi Liu}, \bibinfo{person}{Victor~S. Sheng},
  \bibinfo{person}{Jiajie Xu}, \bibinfo{person}{Fuzhen Zhuang},
  \bibinfo{person}{Junhua Fang}, {and} \bibinfo{person}{Xiaofang Zhou}.}
  \bibinfo{year}{2019}\natexlab{}.
\newblock \showarticletitle{Graph Contextualized Self-Attention Network for
  Session-based Recommendation}. In \bibinfo{booktitle}{\emph{IJCAI}}.
\newblock


\bibitem[\protect\citeauthoryear{Ying, He, Chen, Eksombatchai, Hamilton, and
  Leskovec}{Ying et~al\mbox{.}}{2018}]%
        {graph-rex18}
\bibfield{author}{\bibinfo{person}{Rex Ying}, \bibinfo{person}{Ruining He},
  \bibinfo{person}{Kaifeng Chen}, \bibinfo{person}{Pong Eksombatchai},
  \bibinfo{person}{William~L. Hamilton}, {and} \bibinfo{person}{Jure
  Leskovec}.} \bibinfo{year}{2018}\natexlab{}.
\newblock \showarticletitle{Graph Convolutional Neural Networks for Web-Scale
  Recommender Systems}. In \bibinfo{booktitle}{\emph{KDD}}.
\newblock


\bibitem[\protect\citeauthoryear{Yuan, Karatzoglou, Arapakis, Jose, and
  He}{Yuan et~al\mbox{.}}{2018}]%
        {seq-yuan19}
\bibfield{author}{\bibinfo{person}{Fajie Yuan}, \bibinfo{person}{Alexandros
  Karatzoglou}, \bibinfo{person}{Ioannis Arapakis}, \bibinfo{person}{Joemon~M.
  Jose}, {and} \bibinfo{person}{Xiangnan He}.} \bibinfo{year}{2018}\natexlab{}.
\newblock \showarticletitle{A Simple but Hard-to-Beat Baseline for
  Session-based Recommendations}.
\newblock \bibinfo{journal}{\emph{arXiv 1808.05163}} (\bibinfo{year}{2018}).
\newblock


\bibitem[\protect\citeauthoryear{Zhang, Zhang, Xia, and Sun}{Zhang
  et~al\mbox{.}}{2020b}]%
        {zhang2020graphbert}
\bibfield{author}{\bibinfo{person}{Jiawei Zhang}, \bibinfo{person}{Haopeng
  Zhang}, \bibinfo{person}{Congying Xia}, {and} \bibinfo{person}{Li Sun}.}
  \bibinfo{year}{2020}\natexlab{b}.
\newblock \showarticletitle{Graph-Bert: Only Attention is Needed for Learning
  Graph Representations}.
\newblock \bibinfo{journal}{\emph{arXiv Preprint 2001.05140}}
  (\bibinfo{year}{2020}).
\newblock


\bibitem[\protect\citeauthoryear{Zhang, Liu, Han, Miao, Cui, Li, and
  Tang}{Zhang et~al\mbox{.}}{2020a}]%
        {bpr1}
\bibfield{author}{\bibinfo{person}{Yinan Zhang}, \bibinfo{person}{Yong Liu},
  \bibinfo{person}{Peng Han}, \bibinfo{person}{Chunyan Miao},
  \bibinfo{person}{Lizhen Cui}, \bibinfo{person}{Baoli Li}, {and}
  \bibinfo{person}{Haihong Tang}.} \bibinfo{year}{2020}\natexlab{a}.
\newblock \showarticletitle{Learning Personalized Itemset Mapping for
  Cross-Domain Recommendation}. In \bibinfo{booktitle}{\emph{IJCAI}}.
\newblock


\bibitem[\protect\citeauthoryear{Zheng, Lu, Jiang, Zhang, and Yu}{Zheng
  et~al\mbox{.}}{2018}]%
        {spectral_Zheng_2018}
\bibfield{author}{\bibinfo{person}{Lei Zheng}, \bibinfo{person}{Chun-Ta Lu},
  \bibinfo{person}{Fei Jiang}, \bibinfo{person}{Jiawei Zhang}, {and}
  \bibinfo{person}{Philip~S. Yu}.} \bibinfo{year}{2018}\natexlab{}.
\newblock \showarticletitle{Spectral collaborative filtering}. In
  \bibinfo{booktitle}{\emph{RecSys}}.
\newblock


\bibitem[\protect\citeauthoryear{Zhu, Wang, Chen, Liu, and Zheng}{Zhu
  et~al\mbox{.}}{2020}]%
        {Zhu2020}
\bibfield{author}{\bibinfo{person}{Feng Zhu}, \bibinfo{person}{Yan Wang},
  \bibinfo{person}{Chaochao Chen}, \bibinfo{person}{Guanfeng Liu}, {and}
  \bibinfo{person}{Xiaolin Zheng}.} \bibinfo{year}{2020}\natexlab{}.
\newblock \showarticletitle{A Graphical and Attentional Framework for
  Dual-Target Cross-Domain Recommendation}. In
  \bibinfo{booktitle}{\emph{IJCAI}}.
\newblock


\end{thebibliography}

\appendix

\section{The Optimization of \sysname}

\sysname{} solves the following minimization problem using eigendecomposition of $L_{\text{reg}} = (1-\alpha)D-W+\alpha d_{\text{max}} I$. 
\begin{equation}
\begin{split}
    \eigv[k] = & \argmin_{\eigv} \frac{1}{2} \sum_{i} \sum_{j} W_{i,j} (\bm{q}_{i} - \bm{q}_{j})^2 \\ 
    + & \alpha \sum_{i} (d_{\text{max}} - d_i) q_i^2\\
    = & \argmin_{\eigv} \eigv^\top L_{\text{reg}} \eigv \\
    \text{s.t. \, } & \eigv^\top \eigv = 1 \text{\, and \,} \bm{q}^\top \eigv[l] = 0, \forall l < k.
\end{split}
\label{eq:lepor-optimization-matrix}
\end{equation}
In this section, we closely examine the rationale behind this approach. 

We start by examining Laplacian Eigenmaps, where
the eigenvectors $\eigv[1], \ldots, \eigv[\mathcal{N}]$ can be thought of as solutions of the constrained minimization 
\begin{equation}
\begin{split}
    \eigv[k] = & \argmin_{\eigv} \eigv^\top L \eigv \\
    = & \argmin_{\eigv} \eigv^\top (D-W) \eigv \\
    = & \argmin_{\eigv} \frac{1}{2} \sum_{i}\sum_{j} W_{i,j} (\bm{q}_{i} - \bm{q}_{j})^2 \\
    \text{s.t. \, } & \eigv^\top \eigv = 1 \text{\, and \,} \bm{q}^\top \eigv[l] = 0, \forall l < k.
\end{split}
\label{eq:le-optimization2}
\end{equation}

\begin{lemma}
The eigenvectors of the graph Laplacian, $L = D - W$, are the solutions of the constrained minimization problem described by Eq. \ref{eq:le-optimization2}.
\label{lemma:le}
\end{lemma}

\begin{proof}
We aim to find a sequence of solutions $\bm{y}^{(1)}, \ldots, \bm{y}^{(\mathcal{N})}$ as solutions to the constrained optimization problem described by Eq. \ref{eq:le-optimization2}. 

Since the graph Laplacian $L$ is positive semi-definite, it has real-valued eigenvalues ($0 \leq \lambda_1 \leq \lambda_2 \leq \cdots \leq \lambda_{\mathcal{N}}$). $L$ is also symmetric, so we can perform the eigendecomposition
\begin{equation}
    L = Q^\top \Lambda Q,
    \label{eq:decomp}
\end{equation}
where $\Lambda = \text{diag}(\lambda_1, \lambda_2 , \cdots , \lambda_{\mathcal{N}})$ is a diagonal matrix with the eigenvalues on its diagonal. $Q$ is the orthogonal matrix composed of eigenvectors, $\eigv[1], \ldots, \eigv[\mathcal{N}]$, as its rows, and $Q^{-1} = Q^\top$. 

We can denote the minimization objective as $R(\bm{x})$, 
\begin{equation}
    R(\bm{x}) = \bm{x}^\top L \bm{x} = (Q\bm{x})^\top \Lambda (Q\bm{x}).
\end{equation}
For any $\bm{y} \in \mathbb{R}^{\mathcal{N}}$, we can find $\bm{x} = Q\bm{y}$ such that $Q^\top\bm{x} = \bm{y}$:
\begin{equation}
\begin{aligned}
    R(\bm{y}) &= (Q\bm{y})^\top \Lambda (Q\bm{y}) \\&= (QQ^\top\bm{x})^\top \Lambda (QQ^\top\bm{x}) \\&= \bm{x}^\top \Lambda \bm{x} \\&= \lambda_1\bm{x}_1^2 + \lambda_2\bm{x}_2^2 + \cdots + \lambda_n\bm{x}_n^2.
\end{aligned}
\label{eq:eigen-form}
\end{equation}

From the constraint $\bm{y}^\top \bm{y} = 1$, it is easy to see that $\bm{x}^\top \bm{x} = 1$. Recall that the eigenvalues are arranged in ascending order
\begin{equation}
0 \leq \lambda_1 \leq \lambda_2 \leq \ldots \leq \lambda_{\mathcal{N}}. 
\end{equation}
To minimize $R(\bm{y})$, we just need to allocate the norm of $\bm{x}$ entirely to the its first component, resulting in $\bm{x}^{(1)} = [1, 0, \ldots, 0]^\top$. Hence, the corresponding $\bm{y}^{(1)} = Q^\top \bm{x}^{(1)}$ is the first eigenvector $\eigv[1]$. The minimum is $R(\bm{y}^{(1)}) = \lambda_1$. 

From the second solution onward, we have the additional constraints that $\bm{y}^{(k)\top} \bm{y}^{(l)} = 0, \forall l < k$. Plugging in $\bm{y}^{(k)} = Q^\top \bm{x}^{(k)}$, we can see orthogonality in $\bm{y}$ implies orthogonality in $\bm{x}$: $\bm{x}^{(k)\top} \bm{x}^{(l)} = 0, \forall l < k$. To find $\bm{x}^{(2)}$ that is orthogonal to the first solution $\bm{x}^{(1)}$, we just need to set its first component to 0. Hence, the minimizing solution $\bm{x}^{(2)} = [0, 1, 0, \ldots, 0]^\top$ and $\bm{y}^{(2)} = \eigv[2]$.

Repeating the above steps, we can show the eigenvectors $\eigv[1], \ldots, \eigv[\mathcal{N}]$ are the solutions to the constrained minimization problem. \footnote{This is a widely known proof and we do not make any claim of originality.}
\end{proof}

\begin{theorem}
The eigenvectors of the regularized graph Laplacian, $L_{\text{reg}}$, are the solutions of the constrained minimization problem described by Eq. \ref{eq:lepor-optimization-matrix}.
\end{theorem}

\begin{proof}
Because $d_{max} = \max (d_i \in D)$, the regularized graph Laplacian, $L_{\text{reg}}$, remains positive semi-definite. Applying the same proof in Lemma \ref{lemma:le}, the eigenvectors of $L_{\text{reg}}$ are the solutions for the constrained minimization problem described by Eq. \ref{eq:lepor-optimization-matrix}.
% For an undirected graph, its adjacency matrix $W$ is symmetric. By the definition of degree metrix $D$, the defined regularized graph Laplacian, $L_{\text{reg}}$, is symmetric. 
\end{proof}

% \section{Hyperparameter Sensitivity: Regularization Coefficient}

% We investigate the sensitivity of the population-based regularization coefficient $\alpha$ on \textit{MovieLens-1M} dataset. As shown in Figure~\ref{fig:reg}, we vary the regularization coefficient $\alpha$ in $\{0.1, 0.3, 0.5, 0.7, 1.0\}$, where a larger $\alpha$ refers to stronger regularization. We keep other parameters unchanged (e.g., nearest neighbor $K=1000$, embedding size $D=64$, initial learning rate = 0.001).
% We observe that better performance is achieved by using a relatively larger coefficient (e.g., $\alpha=0.7$). % TODO 

% \section{Reproducibility: Setup Details}
% In order to create the candidate items for each user during testing, we use all items that this user have not interacted with in the training set or the validation. 

\begin{table*}[t]
% \small
\centering
\caption{Performances of different recommendation architectures. The best results are in bold faces and the second best results are underlined.}
\label{exp:difarc}
\begin{tabular}{lllllllllllll}
\toprule
\multirow{2}{*}{Methods} & \multicolumn{4}{c}{Dataset: \textit{ML-1M}} & \multicolumn{4}{c}{Dataset: \textit{Steam}} & \multicolumn{4}{c}{Dataset: \textit{Anime}}\\
\cmidrule{2-13}
& HR@5 & HR@10 & F1@5 & F1@10    & HR@5 & HR@10 & F1@5 & F1@10 & HR@5 & HR@10 & F1@5 & F1@10 \\ \midrule
NO conv   & 0.4331          & 0.5781          & 0.0539          & 0.0815          & 0.1729          & 0.2667          & 0.0301          & 0.0380          & 0.4627          & 0.6081          & 0.0663          & 0.0935          \\
ONLY conv & 0.5088          & 0.6728          & 0.0718          & 0.1101          & 0.1777          & 0.2755          & 0.0313          & 0.0395          & 0.5835          & 0.7187          & 0.0861          & 0.1259          \\
ResNet    & \textbf{0.5406} & \textbf{0.7026} & {\ul 0.0755}    & {\ul 0.1158}    & {\ul 0.1958}    & {\ul 0.2994}    & {\ul 0.0352}    & {\ul 0.0437}    & \textbf{0.6361} & \textbf{0.7654} & {\ul 0.1010}    & {\ul 0.1421}    \\
RevNet    & {\ul 0.5394}    & {\ul 0.6965}    & \textbf{0.0766} & \textbf{0.1162} & \textbf{0.1962} & \textbf{0.3023} & \textbf{0.0353} & \textbf{0.0443} & {\ul 0.6288}    & {\ul 0.7604}    & \textbf{0.1021} & \textbf{0.1442}\\
\bottomrule
\end{tabular}
\end{table*}

\section{Performances of Different Architecture Designs in \modname{}}
Table \ref{exp:difarc} summarizes the experimental results of different architecture designs in \modname{}, on \textit{ML-1M}, \textit{Steam} and \textit{Anime} datsets.

\section{Reproducibility: Time Cost}
For eigen decomposition of both LE and \sysname{} initialization schemes, we use the package of ``scipy.sparse.linalg.eigs'', which is an implementation of IRLM method \cite{irlm,hou2018fast}. However, IRLM method is proposed for computing the first several largest eigenvalues. We shift the initial matrix $W$ using  $W' = \lambda I - W$, where $\lambda = \max (W)$. The result is averaged over 3 repeated experiments.

\section{Reproducibility: Random Seed}
Readers may replicate the results reported in the paper using the random seed of 123. 

\section{Reproducibility: The Simulation in Figure \ref{fig:embedding-change}}

The simulation aims to find the change in the embeddings of nodes when one additional edge is attached to it. To do so, it gathers statistics from 50 random Barabasi-Albert graphs \cite{BA-graph2002} using the NetworkX library.\footnote{https://networkx.github.io/} The algorithm starts from an initial graph with $m_0$ nodes and no edges and gradually adds nodes and random edges until the graph contains $n$ nodes. Here $m_0 = 10$ and $n=100$.

In every random graph, we pick one random node and add a random edge to it. We compute 40-dimensional embedding for that node before and after the addition using Laplacian Eigenmaps. We then compute the $\ell_2$ norm of the change in the embeddings. We perform 30 independent edge insertions for every node in the graph, and record the degree of the node before the addition and the average $\ell_2$ norm. The statistics are averaged again over 50 random graphs. 

Note the random Barabasi-Albert graphs thus generated have very few nodes with degrees less than 10 or greater than 60. That is why the x-axis in Figure 1 starts at 10 and ends at 60. 

\section{Reproducibility: Experimental Settings of Initialization Schemes}
\label{appendix:initialization_setting}
The output embedding size for all initialization schemes is set to $64$. Specifically, in BPR \cite{bpr}, regularization coefficients for both users and items are set to $0.1$ and the coefficient for gradient update is set to $0.25$. NetMF\footnote{https://github.com/xptree/NetMF} \cite{netmf} is implemented using a large window size (i.e., 10), which usually performs better than using a small window size (i.e., 1). In node2vec\footnote{https://github.com/aditya-grover/node2vec} \cite{node2vec} initialization scheme, the length of walking per source is set to $80$ and number of walks is set to $10$. In Graph-Bert\footnote{https://github.com/jwzhanggy/Graph-Bert} \cite{zhang2020graphbert}, we utilize the network structure recovery as the pre-training task since the node attributes are randomly generated in our experiments.

In the KNN graphs, which is the input for NetMF, LE and \sysname{}, the number of nearest neighbors $K$ is set to 1000. The regularization coefficient $\alpha$ in \sysname{} is set to $0.5$.

\section{Reproducibility: Experimental Settings of Recommendation Methods}
\label{appendix:setting}
The embedding size for all neural recommendation models (i.e., NCF\footnote{https://github.com/yihong-chen/neural-collaborative-filtering} \cite{ncf}, NGCF\footnote{https://github.com/xiangwang1223/neural\_graph\_collaborative\_filtering} \cite{NGCF19}, DGCF\footnote{https://github.com/xiangwang1223/disentangled\_graph\_collaborative\_filtering} \cite{DGCF19}, HGN\footnote{https://github.com/allenjack/HGN} \cite{hgn} and our proposed \modname{}) is set to $64$. Other hyperparameters like learning rate and the number of training steps are tuned by grid search on the validation set. The initial learning rate is selected from the range $[0.0001, 0.01]$ and weight decay is selected from the range $[0, 0.01]$. We use Adam \cite{kingma2014adam} for training all neural recommendation methods.

Specifically, in \modname{}, the number of recent items, $l$, is set to $5$. In Feature Network, we employ two similar residual connections. In the first residual block, we have two convolutional layers (kernel size = 3, stride = 1, output channel size = 64), activated by ReLU. In the second block, we employ convolutional layers (kernel size = 3, stride = 2, output channel size = 128) and add downsampling (convolutional layers where kernel size = 1, stride = 2) to the initial input to match the size.

The Generative Network $\mathcal{G}$ and Discriminative Network $\mathcal{S}$ are implemented with the same network structure except the final layer, which consists of three fully connected neural layers (output sizes of each layer are $[256, 256, 128]$ respectively), activated by ReLU. In $\mathcal{G}$, we use Tanh as the activation function in the final layer (output size = $64$). $\mathcal{S}$ is activated by ReLU in the final layer (output size = $1$). We set both margins in Eq. \ref{loss:discriminative} and Eq. \ref{loss:generative} as 0, that is, $m_S = m_G =0$. We use the sum of $\mathcal{L}_S$ and $\mathcal{L}_G$ for training and select the Discriminative Network for evaluation in Table \ref{exp:all1} and Table \ref{exp:all2}.

Other traditional recommendation methods (i.e., TopPop, UserKNN \cite{userknn}, ItemKNN \cite{itemknn}, SLIM \cite{slim})\footnote{https://github.com/MaurizioFD/RecSys2019\_DeepLearning\_Evaluation} are implemented following the work of \cite{recsys19}.

\section{Reproducibility: Experimental Settings of Different Architecture Designs}
In Section \ref{sec:difarch}, which explores the effect of residual connections in Feature Network in \modname{}, the settings of different architectures are as follows. ``No Conv'' replace the residual connections with four fully connected layers, where the output sizes for each layer are[256, 256, 256, 256]. ``Only Conv'' includes simple convolutional layers and other parameters are the same as the ``ResNet'', which is described in Appendix \ref{appendix:setting}. ``RevNet'' replaces previous residual connections with modified ones. It first splits the input $x$ channel-wise into two equal parts $x_1$ and $x_2$. The output is then the concatenation of $y_1 = x_1 + F(x_2)$ and $y2 = x2 + G(y_1)$. Both $F(\cdot)$ and $G(\cdot)$ includes two convolutional layers, activated by ReLU. 
We add additional one convolutional layer (kernel size = 3, stride = 2, output channel size = 128) between the two modified residual blocks and other parameters are the same to the ``ResNet'' architecture. 

\section{Reproducibility: Experimental Settings of Hyperparameter Study}
In Section \ref{sec:sensitivity}, which study the sensitivity of hyperparameters in building \sysname{} initialization, we keep other parameters unchanged. In other words, we set nearest neighbor $K=1000$ in Figure \ref{fig:alpha} and set regularization coefficient $\alpha=0.5$ in Figure \ref{fig:neighbor}.

\end{document}